\documentclass[a4paper,12pt]{article}
\usepackage{amsfonts, amsthm, amsmath, amssymb}
\usepackage{natbib}
\usepackage{threeparttable}
\usepackage{fullpage,amsmath}
\usepackage{setspace}
\newcommand{\dd}{\mathrm{d}}
\newcommand{\dx}{\mathrm{d}x}
\newcommand{\dy}{\mathrm{d}y}
\newcommand{\N}{{\bf N}}
\newcommand{\dxi}{\mathrm{d}\xi}
\newcommand{\bs}[1]{{\bf #1}}
\newcommand{\sk}[1][k]{{\setminus #1}}
\newcommand{\pin}{\widetilde{\pi}^N}
\newcommand{\E}{\mathbb{E}}
\newcommand{\x}{{\bf x}}
\newcommand{\wh}{\widehat}
\newcommand{\wt}{\widetilde}

\newcommand{\dcolon}{\colon\hspace{-2.5pt}}

\newtheorem{theorem}{Theorem}
\newtheorem{definition}{Definition}
\newtheorem{assumption}{Assumption}
\newtheorem{lemma}{Lemma}

\newtheorem{corollary}{Corollary}

\newtheorem{remark}{Remark}
\newtheorem{algorithm}{Algorithm}

\providecommand{\keywords}[1]{\par\noindent{\small{\em Keywords\/}: #1}}
\title{Markov Interacting Importance Samplers}
\author{Eduardo F. Mendes \and Marcel Scharth \and Robert Kohn}
\date{\small{\today}}
\begin{document}
\maketitle
\begin{abstract} We introduce a new Markov chain Monte Carlo (MCMC) sampler called the Markov Interacting Importance Sampler (MIIS). The MIIS sampler uses conditional importance sampling (IS) approximations to jointly sample the current state of the Markov Chain and estimate conditional expectations, possibly by incorporating a full range of variance reduction techniques.  We compute Rao-Blackwellized estimates based on the conditional expectations to construct control variates for estimating expectations under the target distribution. The control variates are particularly efficient when there are substantial correlations between the variables in the target distribution, a challenging setting for MCMC. An important motivating application of MIIS occurs when the exact Gibbs sampler is not available because it is infeasible to directly simulate from the conditional distributions. In this case the MIIS method can be more efficient than a Metropolis-within-Gibbs approach. We also introduce the MIIS random walk algorithm, designed to accelerate convergence and improve upon the computational efficiency of standard random walk samplers. Simulated and empirical illustrations for Bayesian analysis show that the method significantly reduces the variance of  Monte Carlo estimates compared to standard MCMC approaches, at equivalent implementation and computational effort.
\bigskip

\end{abstract}

\clearpage

\section{Introduction}
This paper introduces \emph{Markov interacting importance samplers} (MIIS), a general Markov Chain Monte Carlo (MCMC) algorithm that iterates by sampling the current state from a conditional importance sampling approximation to a target distribution. An importance sampling (IS) approximation consists of a set of weighted samples from a proposal distribution that approximates the target.  Markov interacting importance samplers are conditional in the sense that the importance distribution may depend on the previous state of the Markov chain. The marginal distribution of the states converges to the target distribution for any number of importance samples at each iteration of the Markov chain; the algorithm does not induce an approximation error.

We adopt importance sampling as a basic tool from the perspective that it can be more efficient than a Metropolis-Hastings sampler based on an identical proposal. Importance sampling naturally incorporates the information from all generated samples, while standard Metropolis-Hastings estimates lose information from rejected draws. In addition, importance sampling estimates are based on independent samples and as a consequence the method is immediately amenable to a range of variance reduction techniques (such as antithetic sampling and stratified mixture sampling), as well as convenient to implement and parallelize. It is not standard practice in applied work to incorporate these features into Metropolis-Hastings approaches as they are more challenging to design and use efficiently in an MCMC framework. See for example \citet{craiu2007}, \citet{hammer2008}, \citet{jacob2011}, and \citet{dk2012}.

Importance sampling can be efficient when we are able to construct numerically accurate and computationally fast approximations to a full target distribution. \citet{ZR2007}, \citet{hod2012} and \citet{lietal2013} are recent contributions in this area that have led to the application of IS to challenging problems: see for example \citet{lrv2013} and \citet{tspk2013}. We motivate MIIS by observing that even if the joint target density is intractable by global approximation, we can frequently obtain efficient importance samplers for the conditional distributions. MCMC methods provide a natural way of handling large dimensional problems by sampling from conditional distributions (Gibbs sampling) or by generating samples from complex target densities through local exploration. The MIIS algorithm leverages the advantages of importance sampling in this setting.

As a leading application, we consider the case in which it is not possible to implement an exact Gibbs sampler due to infeasibility of direct simulation from the conditional distributions. The MIIS method relies on IS approximations of the conditional distributions to sample the current state of the Markov Chain. The advantage of importance sampling is that we can additionally use the approximation (that is, all the generated samples) to estimate conditional expectations, possibly by incorporating the full range of variance reduction methods available for standard importance sampling. We compute Rao-Blackwellized estimates based on the conditional expectations to construct control variates for estimating expectations under the target distribution. The control variates are particularly effective when there are substantial correlations between the variables in the target distribution. This is a challenging setting for standard MCMC approaches because the conditioning scheme may imply strong serial correlation in the Markov chain.

We introduce the general MIIS algorithm and present four examples that demonstrate its flexibility. The first two examples present the implementation of MIIS based on simple importance sampling targeting the full and conditional distributions. We derive conditions for the ergodicity and uniform ergodicity of the sampler. The third example introduces antithetic variables and is also uniformly ergodic under general conditions. The final example introduces the MIIS random walk algorithm, designed to accelerate convergence and improve upon the computational efficiency of standard random walk samplers. The random walk sampler is uniformly ergodic assuming that the importance weights are bounded. Ergodicity holds under milder constraints.


Our method relates to the Particle Gibbs (PG) algorithm developed for Bayesian inference in general state space models by \citet{andrieuetal2010}. The PG algorithm iteratively draws the latent state trajectories from its high-dimensional smoothing distribution using a particle filter approximation, and the parameters of the model from their conditionals given the state trajectories. \citet{lindstenschon2012}, \citet{lindstenetal2014as}, \citet{mendesetal2014} and \citet{carteretal2014} present extensions, while  \cite{chopinsingh2013}, \citet{andrieuetal2013} and \citet{lindstenetal2014} study the theoretical aspects of the algorithm. We can show that the particle Gibbs algorithm is a particular type of MIIS. Compared to PG, the MIIS algorithm addresses a wider class of sampling problems and the use of variance reduction methods.

We illustrate Markov interacting importance samplers in a range of examples. We consider the estimation of the posterior mean for a Bayesian Mixed Logit model using the health dataset studied by \citet{fklw2010}. The presence of unobserved heterogeneous preferences in this discrete choice model motivates the use of MCMC methods that iteratively sample the model parameters and the latent choice attribute weights conditional on each other. The results show that the MIIS algorithm with control variates increases efficiency in mean squared error by a factor of four to twenty compared to the Metropolis-within-Gibbs algorithm, which is a standard tool for problems that are not amenable to exact Gibbs sampling. We also implement the MIIS random walk importance sampler for carrying out posterior inference for Markov modulated Poisson processes, a problem considered for example by \citet{fs2006}. Our analysis reveals four to hundredfold gains in efficiency over the standard random walk Metropolis algorithm and the multiple-try Metropolis algorithm of \citet{llw2000}. In this context, the improvements are mainly due to parallelization and better convergence of the Markov chain.



\section{Markov Interacting Importance Samplers}\label{MIIS}
To focus on the main ideas, we use densities in our mathematical discussion up to Section \ref{illustrations}. We assume that the densities are defined with respect to measures that we leave unspecified for now. We provide a more precise treatment in Section~\ref{theory} and the appendix.

\subsection{Notation and basic definitions}\label{ss:notation}
This subsection presents some of the notation used in the article. We define the basic random variables on a set $A$ that is a subset of Euclidean space. Suppose that $f(x)$ is a real function with $x \in A$. We take any density $\nu(x)$ on $A$ to be with respect to some measure on $A$, which we denote as $\dx$. We define the expected value of $f$ with respect to the density $\nu $ as
\begin{align} \label{eq: expect def}
E_\nu(f) & := \int f(x) \nu (x)dx\,
\end{align}
provided the integral exists.

In our article, $\pi(x)$ is the target density. We often can evaluate $\pi(x)$ only up to a constant of proportionality $m(x)$, with $\pi(x) = m(x)/Z_m$, where $Z_m=\int_A m(x) dx $ is the normalizing constant. Suppose that $x_i \in A, i=1,\dots, N$. Then, for $1\le i\leq j \le N $, we define $i \dcolon j := \{i,i+1,\dots,j\}$, $x_{i:j} := (x_i,\dots,x_j)$ and $x_{\sk[k]} := (x_1,\dots,x_{k-1},x_{k+1},\dots,x_N)$.

\subsection{Conditional Importance Sampler}  \label{SS: CIS}

This section introduces the \textit{conditional importance sampler} (CIS) which is the basic building block of the MCMC algorithms in this article. The CIS is motivated by the question:
\lq \lq \textit{how to implement an importance sampler approximation to $\pi$ that provides unbiased samples?}\rq \rq
 The CIS is our solution to this problem. We go beyond simple importance sampler and construct a general framework that not only covers the simple importance sampling approximation with variance reduction techniques, but also extends the basic importance sampling paradigm, allowing local exploration of the target inside an MCMC setting, for instance, by using a random-walk approach.

At each iterate of an MCMC algorithm, the CIS constructs an empirical approximation to the target density $\pi(\cdot)$. It generates an auxiliary variable $\xi$ and $N$ particles $X_{1:N}$ conditional on the previous iterate $y$, in such a way that one particle $X_k$ is generated through a Markov transition kernel and the other $N-1$ particles are generated conditional on $X_k$.

We now present a more precise description of the CIS. Let $\eta(\xi|y)$ be the conditional density of the auxiliary variable $\xi$, with $\xi, y \in A$, and take $\eta(\xi) = \int \eta(\xi|y) \pi(y)dy$ so that $\pi(y|\xi) = \eta(\xi|y)\pi(y)/\eta(\xi)$. Let $T(y,x;\xi)$ be the density of a Markov transition kernel from $y$ to $x\in A$, conditional on $\xi$, that is reversible with respect to $\pi(y|\xi)$; i.e., $\pi(y|\xi) T(y,x;\xi) = \pi(x|\xi) T(x,y; \xi)$, or equivalently,
 \begin{align}\label{eq: T reversibility}
 \pi(y) \eta(\xi|y)T(y,x; \xi)  & =  \pi(x) \eta(\xi|x)T(x,y; \xi).
 \end{align}
 Given $\xi\in A$, let $\bs{q}(x_{1:N}|\xi)$ be a joint importance distribution with marginals $q_i(x_i|\xi)$ ($i\,=\,1,\dots,N$).
For any $1\le k\le N$, define the conditional density
\begin{equation} \label{eq: cond density q}
\bs{q}_{\sk}(x_{\sk}|x_k,\xi) := \frac{\bs{q}(x_{1:N}|\xi)}{q_k(x_k|\xi)}\, .
\end{equation}

\begin{definition}[Conditional Importance Sampler] \label{def: CIS}
	For any given $y \in A$ and $1\le k\le N$, the \emph{Conditional Importance Sampler} generates $X_{1:N},\xi|(y,k)$ from the probability distribution
\begin{equation}
    \Gamma^N(x_{1:N}, \xi|y,k) := \eta(\xi|y)T(y,x_k;\xi)\,\bs{q}_{\sk}(x_{\sk}|x_k,\xi).
    \label{eq:CIS}
\end{equation}
\end{definition}

The auxiliary variable $\eta$ introduces dependence in the importance sampling approximation. Moreover, we can often choose the auxiliary density $\eta$ so that $w_i(x;\xi)$ is bounded. For instance, the random-walk importance sampling algorithm chooses $\eta(\xi|x) = q(x|\xi) = \phi(|\xi-x|)$. The weights are $w_i(x;\xi) = m(x)$, which are bounded if $m(x)$ is bounded. The dependence on $\xi$ can be easily dropped if one takes $\eta(\cdot|y) = \eta(\cdot)$ and each $q_i(\cdot|\xi) = q_i(\cdot)$. The Markov transition kernel $T(y,\cdot;\xi)$ can be taken as the identity kernel, i.e., $T(y,\cdot;\xi)= \delta(\cdot-y)$, which is our choice in Sections \ref{examples} and \ref{illustrations}. A Metropolis-Hastings kernel targeting $\pi(\cdot|\xi)$ is also a valid choice.

\vspace{1em}
The CIS generates $(X_{1:N},\xi)$ using the following algorithm.
\begin{algorithm}[Conditional Importance Sampler]\label{alg: CIS algorithm}
Given $(y,k)$,
\begin{enumerate}
	\item sample $\xi\sim\eta(\xi|y)$;
	\item sample $X_k \sim T(y, x_k;\xi)$; i.e., generate the particle $x_k$ using the Markov kernel.
	\item sample $X_{\sk} \sim \bs{q}_{\sk}(x_{\sk}|x_k,\xi)$; i.e., generate all the remaining particles conditional on $\xi$ and the propagated particle $x_k$.
\end{enumerate}
\end{algorithm}
\vspace{1em}

From the output of the Conditional Importance Sampler we define the weights for $i=1,\dots,N$
\begin{equation}\label{eq: wts CIS}
 W_i(x_{1:N}; \xi)  := \frac{w_i(x_i;\xi)}{\sum_{j=1}^Nw_j(x_j;\xi)} \quad \text{where} \quad w_i(x;\xi): = \frac{m(x)}{q_i(x|\xi)}\eta(\xi|x)
\end{equation}
and let ${\wh \pi}^N_{CIS}:=\{(x_1, W_{1}(x_{1:N}, \xi)) ,\dots,  (x_N, W_{N}(x_{1:N},\xi))\}$ be the empirical approximation to $\pi$. The weights depend on the marginals $q_i(\cdot|\xi)$ ($i=1,\dots,N$) of $\bs{q}(x_{1:N}|\xi)$, the auxiliary distribution $\eta(\xi|\cdot)$ and the target distribution $\pi(\cdot)\propto m(\cdot)$. Based on ${\wh \pi}^N_{CIS}$, we define the estimator of $E_\pi(f)$ as
\begin{align}\label{eq: CIS prelim}
	{\wh E}_{CIS}^{N}(f) & := \sum_{i=1}^N W_i(x_{1:N},\xi) f(x_i) = E_{{\wh \pi}_{CIS}^N}(f).
\end{align}

Define the joint density
\begin{align} \label{eq:fullmiis}
	{\wt \pi}^N(k,y,x_{1:N},\xi) := N^{-1}\pi(y) \Gamma^N(x_{1:N}, \xi|y,k).
\end{align}
Lemma~\ref{thm:CISunbiased} gives some fundamental properties of ${\wt \pi}^N(k,y,x_{1:N},\xi)$ and shows that the expectation  of ${\wh E}_{CIS}^{N}(f)$  is $E_\pi(f)$ if the marginal distribution ${\wt \pi}^N(y,k) = N^{-1}\pi(y)$. We use ${\wh E}_{CIS}^{N}(f)$, additively, within an MCMC scheme to construct unbiased estimators of $E_\pi(f)$. The unbiasedness property is critical for the variance reduction techniques in Section \ref{s:VR}.

\begin{theorem}\label{thm:CISunbiased}
	Suppose that $E_\pi(|f|)$ is finite, $(k,y)$ is a sample from $N^{-1}\pi(y)$, and that $(x_{1:N},\xi)$ is generated from $\Gamma^N(x_{1:N},\xi|y,k)$. Then,
\begin{enumerate}
\item [(i)]
${\wt \pi}^N(y) = \pi(y) $.
\item [(ii)]
\begin{align}\label{eq: joint k and x}
		{\wt \pi}^N (k,y|x_{1:N},\xi) & = \sum_{i=1}^N W_i(x_{1:N},\xi)I(k=i)T(x_i, y;\xi),
\intertext{or equivalently,}
{\wt \pi}^N (K=i|x_{1:N},\xi) & =  W_i(x_{1:N},\xi)\quad \text{and} \quad {\wt \pi}^N (y|x_{1:N}\xi,k) = T(x_k, y;\xi).
\end{align}
\item [(iii)]
$E_{{\wt \pi}^N }\Big ( {\wh E}_{CIS}^{N}(f)\Big ) = E_\pi(f)$.
\end{enumerate}
\end{theorem}

\begin{remark}\label{Rem: IS remark}
We now compare importance sampling to conditional importance sampling. In importance sampling, we draw particles $x_{1:N}$ from an \textit{importance} or \textit{proposal} density $\bs{q}(x_{1:N})$ with marginal densities  $q_i(x_i)$ and calculate their \textit{importance weights}
 \begin{equation*}
	 W_i(x_{1:N}):=\frac{w_i(x_i)}{\sum_{j=1}^N w_j(x_j)}, \quad \text{where} \quad w_i(x_i) := \frac{m(x_i)}{q_i(x_i)},
 \end{equation*}
 to obtain the approximation ${\wh \pi}^N_{IS} :=\{W_{1:N}(x_{1:N}),x_{1:N}\}$ to $\pi$. The IS sampling estimate of $E_\pi(f)$ is
\begin{equation} \label{eq: IS estimator}
   {\wh E}^N_{IS}(f) := \sum_{i=1}^N W_i(x_{1:N}) f(x_i) = E_{{\wh \pi}^N_{IS}} (f)
\end{equation}
In the simplest case, the \textit{particles} $x_{1:N}$ are sampled independently from the same proposal distribution $q$, i.e., $q_1=\dots=q_N=q$ and $\bs{q}(x_{1:N}) = \prod_{i=1}^N q(x_i)$. Despite similarities, there fundamental differences between using ${\wh \pi}_{CIS}^{N}$ and ${\wh \pi}_{IS}^{N}$.
\begin{enumerate}
\item
	The marginal distribution of a sample $X$ from ${\wh\pi}_{IS}^N$ is not $\pi(X)$, while the distribution of $Y$ from ${\wh\pi}_{CIS}^N$ is $\pi(Y)$. Similarly,
\begin{align} \label{eq: IS biased}
E_q \Big ( {\wh E}_{\pi}^{IS}(f)\Big ) \neq E_\pi(f),
\end{align}
whereas $E_{{\wt \pi}^N }\Big ( {\wh E}_{CIS}^{N}(f)\Big ) = E_\pi(f)$.
\item The weights $w_i$ in the CIS may depend on an auxiliary variable $\xi$, with density $\eta(\cdot|y)$, that incorporates past information in the proposal opening the possibility for using local proposals. Moreover, it can be used as a mechanism to bound the weights and provide more robust estimators.
\end{enumerate}
\end{remark}

\subsection{Markov Interacting Importance Sampling Algorithm}\label{s:miisfull}

The MIIS algorithm simulates from the target distribution $\pi$ on $A$. It iterates by first constructing a discrete approximation to $\pi$ using the CIS, conditional on the previous state $(y,k)$ of the Markov Chain, and then samples from the approximation. It requires specifying a joint proposal distribution $\bs{q}(x_{1:N}; \xi)$, an auxiliary distribution $\eta(\xi|y)$, and a Markov transition kernel $T(y,x;\xi)$.

\vspace{1em}
\begin{algorithm}[Markov Interacting Importance Sampler]\label{alg: MIIS algorithm}

Given $y^{(0)}\in A$ and $1\le k^{(0)}\le N$, at step $t=1,2,\dots$
\begin{enumerate}
	\item Generate $\xi^{(t)}|y^{(t-1)} \sim \eta(\xi|y^{(t-1)})$.
	\item Generate $X_{k^{(t-1)}}^{(t)}|(y^{(t-1)},\xi^{(t)})\sim T\left( y^{(t-1)},x_{k^{(t-1)}}^{(t)};\xi^{(t)} \right)$.
	\item Generate
        \[
            X_{\sk^{(t-1)}}^{(t)}\Big|\left(x_{k^{(t-1)}}^{(t-1)},k^{(t-1)},\xi^{(t)}\right) \sim\bs{q}_{\sk^{(t-1)}}\left(x_{\sk^{(t-1)}}^{(t)}\Big| x_{k^{(t-1)}}^{(t)},k^{(t-1)},\xi^{(t)}\right).
        \]
	  \item
 For $k=1, \dots, N$, calculate
\begin{equation*}
w_k\left( x_k^{(t)};\xi \right) = \frac{\,m(x_k^{(t)})}{q_k(x_k^{(t)}|\xi^{(t)})}\eta(\xi^{(t)}|x_k^{(t)}), \quad \text{and} \quad W_k(x_{1:N}^{(t)},\xi^{(t)}) = \frac{w_k\left( x_k^{(t)};\xi \right)}{\sum_{j=1}^N w_j\left( x_j^{(t)};\xi \right)}\ .
\end{equation*}

Draw $K^{(t)}=k|(x_{1:N}^{(t)},\xi^{(t)})$ with probability $W_k(x_{1:N}^{(t)},\xi^{(t)})$.

	  \item Generate $Y^{(t)}|(x_{1:N}^{(t)},k^{(t)},\xi^{(t)})\sim T\left( x_{k^{(t)}}^{(t)},x^{(t)};\xi^{(t)} \right)$.
  \end{enumerate}
\end{algorithm}

We divide the algorithm into two blocks. The first block consists of steps 1 to 3 and uses the CIS to draw an approximation to $\pi$. It corresponds to Algorithm \ref{alg: CIS algorithm} in Section \ref{SS: CIS}. The second block consists of steps 4 and 5 and draws an element from this approximation. It corresponds to part (ii) of Theorem~\ref{thm:CISunbiased}.

The MIIS algorithm is a Gibbs sampler on an augmented space that contains all variables sampled in the CIS step, i.e., it is a Gibbs sampler targeting \eqref{eq:fullmiis}. It also follows that if $(k^{t-1}, y^{(t-1)})\sim N^{-1}\pi(\cdot)$, the marginal distribution of $y^{(t)}$ is the original target $\pi$; the MIIS algorithm generates samples from $\pi$ without the approximation error induced by the CIS step.

\begin{theorem}[Target Distribution]
	The \textit{Markov Interacting Importance Sampler} is a Gibbs sampler targeting the augmented density \eqref{eq:fullmiis}
that has $\pi(y)$ as a marginal density.
\label{thm:fullmiis}
\end{theorem}
\section{Examples}\label{examples}
This section illustrates the MIIS methodology in three useful examples. For simplicity, the Markov transition density is set to the identity density, i.e., $T(y,x;\xi) = \delta_y(x)$, which denotes a density in$x$ that integrates to 1 and which is zero exact at $x=y$; we will sometimes write it as $\delta(x-y)$. We do not use the auxiliary variable $\xi$  in the first two examples, which is equivalent to assuming that $\eta(\xi|x) = \eta(\xi)$ and $\bs{q}(x_{1:N}|\xi) = \bs{q}(x_{1:N})$. Section~\ref{SS: Examples theory} gives formal convergence results for all three examples.

\subsection{Simple Importance Sampling} \label{SS: simple IS}
This specification corresponds to the iterated Sampling Importance Resampling algorithm (i-SIR) in \cite{andrieuetal2013}. In importance sampling algorithms we generate particles independently from importance distributions $q_i(x)=q(x)$ ($i=1,\dots,N$), i.e., $X_{1:N}\sim\prod_{i=1}^Nq(x_i)$. Hence $\bs{q}(x_{1:N}|\xi) = \prod_{i=1}^Nq(x_i)$ and
\[
\bs{q}_{\sk}(x_{\sk}|x_k,k,\xi) = \prod_{i\ne k}^Nq(x_i).
\]
The CIS in this case is
\[
\Gamma^N(x_{1:N}\xi|y,k) = \eta(\xi)\delta(y-x_k)\prod_{i\ne k}^Nq(x_i).
\]
Algorithm~\ref{alg: MIIS with indep particle} follows from  Algorithm~\ref{alg: MIIS algorithm}.
\begin{algorithm}[MIIS with Simple Importance Sampling] \label{alg: MIIS with indep particle}
Given $y^{(t-1)}$ and $k^{(t-1)}=k$,
\begin{enumerate}
    \item Generate $X_i^{(t)}\sim q(x)$, for $i=\{1\dcolon N\}\setminus k$, and set $x_k^{(t)}=y^{(t-1)}$.
	 \item Draw $K^{(t)}=k|x_{1:N}^{(t)}$ with probability proportional to $w_k(x_k^{(t)})=m(x_k^{(t)})/q(x_k^{(t)})$.
    \item Set $y^{(t)}=x^{(t)}_{k^{(t)}}$.
\end{enumerate}
\end{algorithm}

\subsection{Importance Sampling with Antithetic Variables}\label{s:anti}
In the importance sampling literature, the method of antithetic variables consists of drawing perfectly negatively correlated particles to reduce the variance of the Monte Carlo estimate. We can use this method within the MIIS framework. The importance sampler with antithetic variables draws the particles in pairs from a proposal distribution. Suppose that $N$ is even. For $k\le N/2$, let $q_k(x_k)$ be the density of $x_k$ with corresponding cumulative distribution function $Q_k(\cdot) $ and let $x_{N/2+k} = Q_k^{-1} ( 1- Q_k(x_k)) $, where $Q_k^{-1}$ is the inverse of $Q_k$. We write the joint density of $x_k, x_{N/2+k} $ as
\begin{align*}
	q_{k,N/2+k}(x_k,x_{N/2+k}) =  q_k(x_k) \delta_{ Q_k^{-1}(1-Q_k(x_k)) }(x_{N/2+k}) .
\end{align*}
The marginals are $q_k(x) = q_{N/2+k}(x)$ and the conditional density of $X_k$ given $x_{N/2+k}$ is $q_k(x_k|x_{N/2+k}) = \delta_{Q_k^{-1}(1-Q_k(x_{N/2+k})) }(x_k)$.
For notational simplicity assume $k\le N/2$. We sample the particle system given $(x_k,k)$ from
    \begin{eqnarray*}
		 \bs{q}_{\sk}(x_{\sk}|x_k,\xi,k) &=& \delta_{Q_k^{-1}(1-Q_k(x_k))  }(x_{N/2+k})\prod_{i\ne k}^{N/2}q_{i,N/2+i}(x_i, x_{N/2+i})\\
        &=& \frac{\prod_{i=1}^{N/2}q_{i,N/2+i}(x_i,x_{N/2+i})}{q_k(x_k)}\\
		  &=& \frac{\bs{q}(x_{1:N})}{q_k(x_k)},
    \end{eqnarray*}
    and the CIS is
	 \[
		 \Gamma^N(x_{1:N}\xi|y,k) = \eta(\xi)\delta_y(x_k) \frac{\prod_{i=1}^{N/2}q_{i,N/2+i}(x_i,x_{N/2+i})}{q_k(x_k)}.
	 \]

\begin{algorithm} [MIIS with Antithetic Variables]\label{alg: miis with antithetic}
Given $y^{(t-1)}$ and $k^{(t-1)}=k$,
\begin{enumerate}
	\item Generate $(X_i^{(t)},X_{N/2+i})\sim q_{i,N/2+i}(x_i,x_{N/2+i})$, for $i=\{1{:}N/2\}\setminus k$.
	\item If $k \leq N/2$, set $x_k^{(t)}=y^{(t-1)}$, and $x_{N/2+k}=Q_k^{-1}(1-Q_k(x_k^{(t)}))$.
If $k>N/2$, set $x_k^{(t)}=y^{(t-1)}$, and $x_{k-N/2}=Q_{k-N/2}^{-1}(1-Q_{k-N/2}(x_k^{(t)}))$.
   \item Draw $K^{(t)}=k|x_{1:N}^{(t)}$ with probability proportional to $m(x_k^{(t)})/ q_k(x_k^{(t)})$.
   \item Set $y^{(t)}=x^{(t)}_{k^{(t)}}$.
\end{enumerate}
\end{algorithm}

\subsection{Random Walk Importance Sampler}\label{s:randomwalk}
    The random walk importance sampler draws particles from a symmetric proposal dependent on its past. The advantage is that the method bounds the weights by construction. The random walk proposal performs local exploration around the auxiliary variable $\xi$, which we sample conditionally on the previous state.

	 Let $q(\cdot|y) = \eta(\cdot|y) = \phi\left( \cdot-y \right)$ denote the proposal functions for $q_i$ and $\eta$. Then
    \[
		 \bs{q}_{\sk}(x_{\sk}|x_k,k,\xi) = \prod_{i\ne k}^N\phi(x_i-\xi)
    \]
    The CIS is
    \[
		 \Gamma^N(x_{1:N},\xi|y,k) = \delta_y(x_k)\phi(\xi-x_k)\prod_{i\ne k}^N\phi(x_i-\xi).
    \]

	The random walk importance sampler bounds the weights if $m(x)$ is bounded.
The sampling algorithm follows from Algorithm~\ref{alg: MIIS algorithm}

\begin{algorithm} [MIIS with Random Walk proposal]\label{alg: miis with RW}
Given $x^{(t-1)}$ and $k^{(t-1)}=k$,
\begin{enumerate}
	\item Generate $\xi^{(t)}|y^{(t-1)} \sim \phi(\xi-x^{(t-1)})$
	\item Generate $X_i^{(t)}\sim \phi(x-\xi^{(t)})$, for $i=\{1\dcolon N\}\setminus k$, and set $x_k^{(t)}=y^{(t-1)}$.
   \item Draw $K^{(t)}=k|x_{1:N}^{(t)}$ with probability proportional to $m(x_k^{(t)})$.
   \item Set $y^{(t)}=x^{(t)}_{k^{(t)}}$.
\end{enumerate}
\end{algorithm}

\section{MIIS Targeting Conditional Distributions}\label{s:miisgibbs}
This section shows how to use use the MIIS algorithm within a Gibbs sampling framework. We use the following notation. Suppose we partition $x \in A $ as $\{x(1), \dots , x(d)\}$. Then, for $1 \leq s \leq t \leq d$, $x(s\dcolon t):= \{x(s), x(s+1), \dots, x(t)\}, x_i(s\dcolon t): = \{ x_i(s), \dots, x_i(t)\}$, etc. We define $A_s := \{x(s)\dcolon  x \in A\}$ and $A_{\sk[s]} := \{ x(\sk[s])\dcolon  x \in A\}$.  For a density $\nu(x)$, $x \in A$, we define the conditional density $\nu_s(x(s)|x(\sk[s]) ) :=\nu(x)/\nu(x(\sk[s]) )$ and the conditional expectation
\begin{align}\label{eq: cond integral}
 E_{\nu_s(\cdot|x(\sk[s]))}(f) := \int_{A_s} f(x)\nu_s(x(s)|x(\sk[s]))\dx(s).
\end{align}

\subsection{Conditional Importance Sampler for conditional distributions} \label{SS: cis for conditional}
The CIS for conditional distributions is similar to the CIS in Section \ref{SS: CIS}, but now targets $\pi_s(x(s)|x({\sk[s]}))$, $s=1,\dots,d$.  Given $y \in A$, $s \in \{1\dcolon d\}$ and $k_s \in \{1\dcolon N\}$, let $\eta_s(\xi(s)|y(s),y(\sk[s]))$ be the density of the auxiliary variable $\xi(s)$, conditional on $y$. Let $T_s(y(s),x_{k_s}(s);\xi(s),y(\sk[s]))$ be a the density of a Markov transition kernel, conditional on $(\xi(s),y(\sk[s]))$, that is reversible with respect to $\pi_s(y(s)|\xi(s),y(\sk[s])) \propto \pi_s(y(s) |y(\sk[s]))\eta_s(\xi(s)|y(s),y(\sk[s]))$.

Given $\xi(s)$ and $y(\sk[s])$, let $\bs{q}_s(x_{1:N}(s)|\xi(s),y(\sk[s]))$ be a joint importance density with marginals $q_{s,i}(x_i(s)|\xi(s),y(\sk[s]))$ ($i=1,\dots,N$), and
\begin{equation} \label{eq: cond cond}
	\bs{q}_{s,\sk_s}(x_{\sk_s}(s) |x_{k_s},\xi(s),y(\sk[s])) :=\frac{\bs{q}_s(x_{1:N}(s)|\xi(s),y(\sk[s]))}{q_{s,k_s}(x_{k_s}(s)|\xi(s),y(\sk[s]))}.
\end{equation}

\begin{definition}[Conditional Importance Sampler for conditional distributions:]
	For $1\le s\le d$, $y \in A$, and $k_s \in \{1\dcolon N\}$, the \emph{Conditional Importance Sampler for conditional distributions} generates $X_{1:N}(s),\xi(s)|(y(s), k_s,y(\sk[s])) $ from the probability distribution
	\begin{multline}	 \label{eq:CIS-conditional}
		\Gamma_s^N(x_{1:N}(s),\xi(s)|y(s), k_s,y(\sk[s]) ) = \eta_s( \xi(s)|y(s),y(\sk[s]))\, T_s(y(s),x_{k_s}(s);\xi(s), y(\sk[s])) \\
		\times \bs{q}_{s,\sk_{s}}(x_{\sk_{s}}(s)|x_{k_s}(s),k_s,\xi(s),y(\sk[s])).
	\end{multline}
\end{definition}

In the CIS for conditional densities, we first generate $\xi(s)$, then we generate $x_{k_s}(s)$ conditional on $\xi(s)$, and finally the remaining particles $x_{\sk_{s}}(s)$ conditional $\xi(s)$ and $x_{k_s}(s)$

Suppose we express the target $\pi_s(x(s)|x(\sk[s])) \propto m_s(x(s)|x(\sk[s]))$, where we can evaluate $m_s(x(s)|x(\sk[s]))$. From the output of the CIS for conditional distributions, we define the weights
\begin{equation}\label{eq: CIS for cond}
W_{s,i}(x_{1:N}(s);\xi(s)|y(\sk[s]))  = \frac{w_{s,i}(x_{i}(s); \xi(s)| y(\sk[s]))}
{\sum_{j=1}^N w_{s,j}(x_{j}(s);\xi(s)|y(\sk[s]))},
\end{equation}
where
\begin{equation}
w_{s,i}(x_{i}(s); \xi(s)|y(\sk[s]))  = \frac{ m_s(x_i(s)|y(\sk[s]))}{q_{s,i}(x_i(s)|\xi(s),y(\sk[s]))}\,\eta_s(\xi(s)|x_i(s),y(\sk[s]))
\end{equation}
and consider ${\wh \pi}^N_{s,CIS}(\cdot|y(\sk[s])):=\{ (W_{s,1},x_{1}(s)), \dots , (W_{s,N},x_{N}(s))\}$ as an empirical approximation of $\pi_s(\cdot |y(\sk[s]))$. Based on ${\wh \pi}^N_{s,CIS}$, we define the estimator of $E_{\pi_s(\cdot |y(\sk[s]))}(f)$ as
\begin{align}\label{eq: con cond IS }
	{\wh E}_{s,CIS}^{N}(f|y(\sk[s])) &:=  \sum_{i=1}^N  W_{s,i}(x_i(s);\xi(s),y(\sk[s])) f(x_i(s), y(\sk[s])) = E_{{\wh \pi}^N_{s,CIS}(\cdot|y(\sk[s]))}(f).
\end{align}

Analogously to the CIS, define the joint density of $(K_s,Y(s),X_{1:N}(s),\xi(s)$ conditional on $Y(\sk[s])$ as
\begin{align} \label{eq:fullmiis  cond}
{\wt \pi}^N_s(k_s,y(s),x_{1:N}(s),\xi(s)|y(\sk[s])) :=  \frac{\pi_s(y(s)|y(\sk[s]))}{N} \Gamma^N_s(x_{1:N}(s), \xi(s)|y(s),k_s,y(\sk[s])).
\end{align}
Lemma~\ref{thm:CISunbiased-conditional} gives some properties of the density \eqref{eq:fullmiis  cond}
  and shows that if $(k_s, y(\sk[s]))$ is generated from $N^{-1} \pi_s(y(s))|y(\sk[s]))$
then the expectation of ${\wh E}_{s,CIS}^{N}(f)$ is $E_{\pi_s(\cdot|y(\sk[s]))}(f)$.

\begin{theorem} \label{thm:CISunbiased-conditional}
	Suppose $(k_s,y(s))$ be a sample from $N^{-1}\pi_s(y(s))|y(\sk[s]))$, and $(x_{1:N}(s),\xi(s))$ a sample from $\Gamma_s^N(x_{1:N}(s),\xi(s)|y(s),k_s,y(\sk[s]))$. Then, conditional on $y(\sk[s])$,
\begin{itemize}
\item [(i)]
${\wt \pi}^N_s(y(s)) = \pi_s(y(s)) $.

\item [(ii)] The conditional density of $k_s,y(s)$ given $x_{1:N}(s),\xi(s)$ is
\begin{align*}
{\wt \pi}^N_s(k_s,y(s)|x_{1:N}(s),\xi(s)) & = W_{s,k_s}  T(x_{s_k}(s),y(s); \xi(s)) \\
\intertext{or equivalently}
{\wt \pi}^N_s (k_s|x_{1:N}(s), \xi(s))  & = W_{s,k_s}\quad \text{ and } \quad \\
{\wt \pi}^N_s (y(s)|x_{1:N}(s), \xi(s), k_s) & = T_s(x_{k_s}(s), y(s); \xi(s)).
\end{align*}

\item [(iii)]
$
E_{{\wt \pi}_s^N(\cdot |y(\sk[s]))} \Big ( {\wh E}_{s,CIS}^{N}(f)
\Big )
  = E_{\pi_s(\cdot|y(\sk[s]))}(f) .
$
\end{itemize}
\end{theorem}

\subsection{The Markov Interacting Importance Sampler within Gibbs} \label{SS: Gibbs MIIS algorithm}
The algorithm extends the MIIS sampler targeting the full density. It simulates sequentially from the conditional distributions $\pi_1(y(1)|y(\sk[1])), \dots, \pi_d(y(d)|y(\sk[d]))$, using the CIS approximation to the conditionals. The method is an alternative to the Metropolis-within-Gibbs algorithm that is is suitable for the application of the variance reduction techniques in Section \ref{s:VR}. The MIIS within Gibbs sampler requires the specification of joint proposal distributions $\{\bs{q}_s(x_{1:N}(s)|\xi(s),y(s), y(\sk[s])\}$, auxiliary distributions $\{\eta_s(\xi(s)|y(s), y(\sk[s]))\}$, and Markov transition kernels $\{T_s(y(s),x_{k_s}(s);\xi(s),y(\sk[s]))\}$, for each $s=1,\dots,d$. The general form of the MIIS Gibbs sampler is given by Algorithm~\ref{alg:gibbs miis}

\vspace{1em}
\begin{algorithm}  [The Markov Interacting Importance Sampler within Gibbbs]\label{alg:gibbs miis}
Given $y^{(0)}\in A$ and $1\le k_s^{(0)}\le N$, $s=1,\dots,d$, the algorithm at step $t=1,2,\dots$, is described as follows,
with all terms conditional on $y^{(t)}(1{:}s{-}1)$ and $y^{(t-1)}(s{+}1{:}d)$.
\begin{enumerate}
	\item[1.] For $s=1,\dots,d$,
		\begin{enumerate}
			\item[1.1.] Generate $\xi^{(t)}(s) \sim \eta_s(\xi(s)|y^{(t-1)}(s))$.
			\item[1.2.] Generate
			\[
	           X_{k_s^{(t-1)}}^{(t)}(s)\sim T\left( y^{(t-1)}(s),x_{k_s^{(t-1)}}^{(t)}(s);\xi^{(t)}(s))\right).
			\]
		\item[1.3.] Generate
            \[
           X_{\sk^{(t-1)}_s}^{(t)}(s) \sim\bs{q}_{s,\sk^{(t-1)_s}_s}\left( x^{(t)}_{\sk^{(t-1)}_s}(s)\Big|x^{(t)}_{k^{(t-1)}_s}(s),k_s^{(t-1)},\xi^{(t)}(s) \right ) ,
			\]
conditional on $x^{(t)}_{k_s^{(t-1)}}(s),k_s^{(t-1)},\xi^{(t)}(s),y^{(t)}(\sk[s])$.
		\item[1.4.] Draw $K_s^{(t)}=k|(x_{1\dcolon N}^{(t)}(s),\xi^{(t)}(s))$ with probability proportional to
			\[
			 w_{s,k}\left( x^{(t)}_{k}(s); \xi^{(t)}(s)\right)  = \frac{\,m_s\left(x_k^{(t)}(s)\right)\eta_s\left(\xi^{(t)}(s)|x_{k}^{(t)}(s)\right)}
{q_{s,k}\left(x_{k}^{(t)}(s)|\xi^{(t)}(s)\right)}.
            \]

	  \item[1.5.] Generate
		  \[
			  Y^{(t)}(s)\sim T\left( x_{k^{(t)}}^{(t)}(s),y^{(t)}(s);\xi^{(t)}(s) \right).
		  \]
	\end{enumerate}
\item[2.] Set $y^{(t)} = (y^{(t)}(1),\dots,y^{(t)}(d))'$.
\end{enumerate}
\end{algorithm}

\vspace{1em}

For each partition $s=1,\dots,d$, the algorithm iterates as in the MIIS algorithm. Steps 1.1 -- 1.3 construct an approximation ${\wh \pi}^N_{s,CIS}$ to $\pi_s(\cdot|y(\sk[s]))$. Steps 1.4 and 1.5 then draw an element from this approximation.
As before, the MIIS for conditional distributions is a Gibbs sampler on an augmented space that contains all variables sampled in the CIS step. It also follows that the marginal distribution of $y^{(t)}$ is the original target $\pi$. Theorem \ref{thm:conditionalmiis} shows the augmented target distribution and that it generates samples from $\pi$.

\begin{theorem}[Target Distribution]
	The \textit{Markov Interacting Importance Sampler} is a Gibbs sampler targeting the augmented distribution  given by
\begin{equation}
	{\wt \pi}^N(y,\xi,x_{1:N}(1),\dots,x_{1:N}(d),k_{1:d}) = \frac{\pi(y)}{N^d}\prod_{s=1}^d\Gamma_s^N(x_{1:N}(s),\xi(s)|y(s),k_s,y(\sk[s])),
    \label{eq:conditionalmiis}
\end{equation}
and has $N^{-d}\pi(y)$ as a marginal distribution of $(k_{1:d},y)$.
\label{thm:conditionalmiis}
\end{theorem}
\subsection{Example: MIIS within Gibbs with Simple Importance Sampling}\label{s:gibbs}
The MIIS sampler takes the conditional distributions in the Gibbs sampler as the target distributions for the conditional importance samplers. Suppose that we use a simple importance sampling algorithm to construct the CIS approximation. Then, for each $s=1,\dots,d$,
\[
	\Gamma_s^N(x_{1:N}(s),\xi(s)|y(s), k_s,y(\sk[s])) = \eta(\xi(s))\delta(y(s)-x_{k}(s)) \prod_{i\ne k_s}^Nq_{s,i}(x_{i}(s)),
\]
for proposal distributions $q_{s,i}(x_i(s)) = q_s(x_i(s))$.

The distribution of the marginal sequence $x^{(t)}$ generated by this algorithm converges to the full target $\pi$ as the number of iterations increases under suitable regularity conditions that are given in Section~\ref{theory}.

Next algorithm follows from Algorithm~\ref{alg:gibbs miis}. Corollary~\ref{l:IS-CIScond} in Section~\ref{ex: gibbs with IS} gives formal convergence result for Algorithm~\ref{alg: example gibbs miis}.
\medmuskip = 0mu
\begin{algorithm}[MIIS for Gibbs Sampler with Simple Importance Sampling] \label{alg: example gibbs miis}
Given $y^{(0)}$ and $k_{1:d}^{(0)}$,
\begin{enumerate}
	\item for $s=1,\dots,d$
		\begin{enumerate}
			\item Generate $X_{i}(s)|y^{(t)}(\sk[s])\sim q_{s}(x_{i}(s))$, for $i=\{1\dcolon N\}\setminus k_s$, and set $x^{(t)}_{k_s}(s)=y^{(t-1)}(s)$.
			\item Draw $K_s^{(t)}=k|(x^{(t)}_{1:N}(s),y^{(t)}(\sk[s]))$ with probability proportional to the weight $(m_s(x^{(t)}_{k}(s)|y^{(t)}(\sk[s]))/q_{s}(x^{(t)}_{k}(s)) $.
    		\item Update $y^{(t)}(s)=x^{(t)}_{k_s^{(t)}}(s)$.
		\end{enumerate}
	\item $t=t+1$
\end{enumerate}
\end{algorithm}

\section{Estimation of expectations using variance reduction methods}\label{s:VR}
Variance reduction techniques play a central role in Monte Carlo integration. We can directly embed variance reduction methods such as antithetic sampling into the conditional importance sampling approximation. This section takes a step further and considers variance reduction methods based on the output of the MIIS algorithm. Suppose that the algorithm targeting $\pi$ runs for $M$ iterations. The simplest estimator of $E_\pi(f)$, which uses only the output $\{x^{(t)}\}$ from the Markov Chain, is
\begin{equation}
{\wh E}^M_{ MC}(f) := \frac{1}{M}\sum_{t=1}^Mf\left(x^{(t)}\right) = E_{{\wh \pi}_{MC}^M}(f)
\label{eq:pimc}
\end{equation}
where ${\wh \pi}_{MC}^M = \{(1/N,x^{(1)}), \dots, (1/N,x^{(M)})\}$.

We can improve efficiency by \textit{reusing all the particles},
constructing \textit{Rao-Blackwellized} estimators, and using \textit{control variates}.
Section~\ref{s:VR theory} shows that all the estimators in this section are consistent under ergodicity.  We assume throughout this section that the chain has reached the stationary distribution before running $M$ iterations of the algorithm. In this case the estimators are also unbiased. In the practical situation where the initialization is arbitrary, the estimators are asymptotically unbiased in $M$ for a fixed $N$.

\subsection{Reusing all the particles}\label{SS: reusing}

The MIIS algorithm constructs an unbiased approximation
\begin{align}
{\wh E}_{CIS, t}^{N} (f):= \sum_{i=1}^N W_i\left(x_{1:N}^{(t)};\xi^{(t)}\right)f\left(x_i^{(t)}\right) \label{eq: def CIS t}
\end{align}
to $E_\pi(f)$ at each iteration $t$ of the Markov chain, after the chain has converged.
The MIIS estimator that averages over the  terms ${\wh E}_{CIS, t}^{N} (f) $ is
\begin{align}
 {\wh E}^{M,N}_{MIIS} (f) & : =  \frac{1}{M}\sum_{t=1}^M {\wh E}_{CIS, t}^{N} (f) = {\wh E}^M_{MC}\Big ({\wh E}_{CIS,t}^{N} (f)\Big )
\label{eq:pimiis}
\end{align}

\subsection{Rao-Blackwellization}\label{s:RB}
The motivation for Rao-Blackwellized estimators is that the variance of $f(x(s))$ is larger than the variance of $E_{\pi_s(\cdot|x(\sk[s]) } (f)$. However, the latter requires knowledge of the conditional expectation in closed form. The MIIS for the Gibbs sampler overcomes this limitation by using an unbiased approximation of the unknown conditional expectation. It follows from Theorem \ref{thm:CISunbiased-conditional} that, at each iteration $t$ of the Markov chain, the term
${\wh E}_{s,CIS}^{N}(f)$
is an unbiased estimator
 of $E_{\pi_s(\cdot|x(\sk[s])}(f)$. For each $s=1,\ldots,d$, define
\begin{align}
{\wh E}^{M,N}_{s,RB}(f) &= \frac1M \sum_{t=1}^M  {\wh E}_{s,CIS,t}^N(f)
  \label{eq:pirb}
\intertext{where}
{\wh E}_{s,CIS,t}^N(f) &=  \sum_{i=1}^N
W_{s,i}\left(x_{1:N}^{(t)}(s); \xi_s^{(t)}|x^{(t)}(\sk[s]) \right)f\left(x_i^{(t)}\right)
\label{eq: def CIS s t}
\end{align}
and  $x_i^{(t)} = \{x_i(s), x^{(t)}(\sk[s])\}$ and $ x^{(t)}(\sk[s])= \{x^{(t)}(1\dcolon s-1), x^{(t-1)}(s+1\dcolon d)\}$.

We define the Rao-Blackwellized  MIIS  estimator for the Gibbs sampler
as the average of  the marginal Rao-Blackwellized estimators in \eqref{eq:pirb},
\begin{align} \label{eq: RB gibbs estimator}
{\wh E}_{MIIS}^{M,N}(f) &  =  \frac{1}{d}\sum_{s=1}^d {\wh E}^{M,N}_{s,RB}(f).
\end{align}
Both the marginal Rao-Blackwellized MIIS estimators ${\wh E}^{M,N}_{s,RB}(f)$ and the
Rao-Blackwellized  MIIS  estimator for the Gibbs sampler
 ${\wh E}_{MIIS}^{M,N}(f)$ are  unbiased
estimators of $E_{\pi}(f)$ and converge to $E_\pi(f)$  with probability one as $M\rightarrow\infty$, for any $N\ge 2$.

\subsection{Control Variates}\label{s:CV}
It is optimal to further combine the simple Monte Carlo estimator and the MIIS estimator. For $j=1, \dots, p$, suppose that $g_j(x)$ is an integrable function with respect to the density $\pi$ and $U(g_j)$ a real function such that $E_{\wt \pi^N} \big ( U(g_j) \big ) = 0 $. Let
$\boldsymbol\kappa = (\kappa_1,\ldots,\kappa_p)$ be a $p\times 1 $ vector of parameters and let $F = f - \sum_{j=1}^p \kappa_j U(g_j) $.
For an optimal choice of $\boldsymbol\kappa$, we would like the variance of the estimate of the posterior mean of $F$ to be smaller than that of $f$. The variables $U(g_i)$ are the \textit{control variates}. The Monte Carlo estimator using $F$ in place of $f$ is studied in many settings;  \citet{mcmc} and \citet{scicomp}, among others, discuss the standard case. Control variates are not commonly used in an MCMC setting
because the Markov sampling scheme makes it more difficult to find suitable candidate control variates with mean zero.

Define ${\wh E}_{CIS,t}^N (g_j)$ similarly to
\eqref{eq: def CIS t} and
\begin{align}
U_t(g_j) & := g_j(x^{(t)}) - {\wh E}_{CIS,t}^N (g_j) \label{eq: U t j def}
\end{align}

Assuming ergodicity, the samples from the MIIS Markov chain are eventually distributed as $\pin$ and $\pin[U_t(g_i)] =  0$ as required. The estimator with control variates is
\begin{eqnarray}
{\wh E}_{CV}^{M,N}(f;{\bf\kappa}) &=& \frac{1}{M}\sum_{t=1}^M\left\{f\left(x^{(t)}\right)-\sum_{j=1}^p\kappa_j\left[g_j(x^{(t)}-{\wh E}_{CIS,t}^{N}(g_j)\right]\right\}\nonumber\\
&=& \frac{1}{M}\sum_{t=1}^M\left\{f\left(x^{(t)}\right)-\sum_{j=1}^p\kappa_jU_t(g_j)\right\}\nonumber\\
&=&{\wh E}_{MC}^M\left[f-\sum_{j=1}^p\kappa_jU(g_j)\right] = {\wh E}_{MC}^M(F).
\end{eqnarray}
An alternative compact notation shows how we combine the previous estimators,
\begin{equation}\label{e:cvestimator}
{\wh E}_{CV}^{M,N}(f;\bs\kappa) = {\wh E}_{MC}^M(f)-\sum_{j=1}^p\kappa_j\left[{\wh E}_{MC}^M(g_j) - {\wh E}_{MIIS}^{M,N}(g_j)\right].
\end{equation}
In a simple case we may have for example $p=1$ and $g_1(x)=f(x)$, which allows us to take advantage of the typically high correlations between the simple MC and MIIS estimators of $E_\pi(f)$.

The optimal choice of coefficients $\boldsymbol\kappa$ (in the sense of minimizing the variance of the estimator) solves the problem of projecting ${\wh E}_{MC}^M(f)$ on $\sum_{j=1}^p\kappa_j{\wh E}^M_{MC}(U(g_j))$. The solution is $\boldsymbol\kappa^* = \Sigma_{UU}^{-1}\Sigma_{Uf}$, where $\Sigma_{UU}= E({\wh E}^M_{MC}(U)\times{\wh E}^M_{MC}(U)')$ and $\Sigma_{Uf} = E({\wh E}^M_{MC}(U)\times {\wh E}_{MC}^M(f))$, where the expectations are with respect to all the random variables generated by a MIIS Markov Chain with $M$ iterations. In our applications we estimate the covariances by using the overlapping batch means method as in \citet{obm2011}.


We can also use control variates in a Gibbs sampler setting. Our estimator generalizes the control variates approach
used by \citet{dk2012}, which {\em only}  applies to exact Gibbs samplers. For a function $f$ and functions $g_{s,j}$ that are integrable with respect to
$\pi$,
\begin{equation}\label{eq:gibbscv}
{\wh E}_{s,CV}^{M,N}(f;\bs\kappa) := {\wh E}_{MC}^M(f)-\sum_{s=1}^d\sum_{j=1}^{p_s}\kappa_{s,j}\left[{\wh E}_{MC}^M(g_{s,j}) - {\wh E}_{s,RB}^{M,N}(g_{s,i})\right].
\end{equation}
We estimate the optimal parameter $\boldsymbol\kappa = \{\kappa_{1,1}, \dots, \kappa_{1,p_1}, \kappa_{2,1}, \dots, \kappa_{2,p_2}, \dots, \kappa_{d,1}, \dots , \kappa_{d,p_d} \}$ as above.

\section{Illustrations} \label{illustrations}

\subsection{Gibbs sampler with importance sampling}

\subsubsection{Sampling from a bivariate normal distribution}

In this example we sample from a simple bivariate normal distribution to compare the performance of the  MIIS sampler with control variates to the Metropolis-within-Gibbs (MwG) sampler in a setting in which the exact Gibbs sampler is available as a reference. \citet{dk2012} adopt this example to illustrate their use of control variates for the Gibbs sampler. The purpose of this example is to show, in a simple setting, that the MIIS sampler with control variates performs well relative to the MwG and Gibbs samplers. We also present results for the Gibbs sampler with control variates as in
\citep{dk2012}, which we regard as the \lq gold standard\rq for this problem.
Beyond this example, we make the important point that
the MIIS and MwG samplers {\em do not} require being able to sample from exact conditional distributions, whereas it is necessary to sample from the exact conditional distributions for the Gibbs sampler.
All the methods are very simple to implement for this example. The target distribution is
\[\pi(x)\propto \exp\left(-\frac{1}{2}x'\Sigma^{-1}x\right),\quad\Sigma = \left[
				\begin{array}{cc}
					1 & \rho\\
					\rho & 1
				\end{array}
				\right]\]
where $\rho \in \{0.25, 0.5, 0.99\}$ represent low, moderate and high correlation.

We are interested in MCMC estimators of the mean, variance, covariance, a tail probability of the marginal distribution of $x(1)$, i.e.,
$E_\pi(X(1))$, $E_\pi(X(1)^2)-E_\pi(X(1))^2$, $E_\pi(X(1)X(2))-E_\pi(X(1))E_\pi(X(2))$, $E_\pi(I[X(1)<-2.32]) = \Pr(X(1) < -2.32)$,

We implement the MIIS algorithm of Section \ref{s:gibbs} (Algorithm 6). We separately consider the standard case and the use of antithetic variables as in Section \ref{s:anti} (Algorithm 4). The importance distribution $q_{s,i}(x_{s,i})$ for the MIIS method is a Student $t$ with 5 degrees of freedom, shifted and rescaled to have the same mean and variance as the target conditional distribution  $\pi_s(x(s)|x(\sk[s]))$. We use the same proposal for the MwG sampler. The number of particles in the IS approximation is $N=50$. To make the Gibbs and MwG algorithms comparable to MIIS, in these methods we sample 50 iterates of $X(1)$ ($X(2)$) conditional on the current state of $X(2)$ ($X(1)$) in the chain.

We use control variates of MIIS as in Section \ref{s:CV}. The estimator is given by \eqref{e:cvestimator}, where we consider at least two control variates for each moment estimate
\begin{align*}
U_1 & =\pi^M_{MC}(f(x(1)))-\pi_{MIIS}^{M,N}(f(x(1))) = M^{-1} \sum_{t=1}^M f(x^{(t)}(1)) - M^{-1} \sum_{t=1}^M\sum_{i=1}^N  W_i(x^{(t)}_{1:N})f(x_i^{(t)}(1))
\end{align*}
with $U_2=\pi^M_{MC}(f(x(2)))-\pi_{MIIS}^{M,N}(f(x(2)))$ expressed similarly.
The control variates are the differences between the standard MCMC estimates and the corresponding Rao Blackwellized MIIS estimates.  We consider additional control variates for estimating the tail probability and $E_\pi(X(1)X(2))$.  For the tail probability, we include the same control variates used for mean estimation. For estimating $E_\pi(X(1)X(2))$, we incorporate the control variates used for estimating the mean and variance.   We apply the overlapping batch means method in \citet{obm2011} to estimate the covariance matrix of the standard estimator \eqref{eq:pimc} and the control variates based on the output of each chain. That allows us to estimate the optimal coefficients for the control variates as described in Section \ref{s:CV}.

Table \ref{tab:comparison1} summarizes the results. We report the estimated mean square error (MSE) relative to the MwG sampler based on 500 independent Markov Chains with 10,000 iterations (after a burn-in period of 1,000 iterations) . The results reveal that when the correlation in the target bivariate normal distribution is pronounced ($\rho=0.99$), the MIIS method with control variates improves the MSEs for estimating the mean, variance, and covariance by 98-99\% compared to the MwG sampler. The control variates efficiently explore the information in the chain and the high correlation between the two variables to reduce variance. The results for the covariance estimators show that the MIIS approach can work well when estimating expectations which involve variables in different blocks of the sampler.  Introducing antithetic variables in the conditional importance sampler leads to a 99.8\% reduction in MSE compared to MwG. Despite the high correlation in the target distribution, the MIIS estimator with antithetic variables takes advantage of the fact that the mean of the proposal is the exact conditional mean. As $\rho$ becomes lower, the MIIS-CV method displays a lower but still large reduction in MSE in comparison to MwG. The table also shows that as in \citet{dk2012}, the use of control variates in the Gibbs sampler is highly efficient. The disadvantage with the Gibbs-CV method is the requirement that the Gibbs sampler is feasible in the first place, whereas the MIIS-CV applies generally. This simulation exercise illustrates that in many situations, accurate estimation of conditional expectations using MIIS will translate into accurate estimation of expectations under the target distribution with the use of control variates.
\begin{table}[!htbp]
\begin{center}
\caption{Bivariate Gaussian simulation -- Monte Carlo MSE of target density expectation estimates relative to MwG.\label{tab:comparison1}}
\begin{threeparttable}
\begin{tabular}{lccccc}
\\
\hline\hline
&  \multicolumn{5}{c}{$\rho=0.99$} \\
& Gibbs & Gibbs-CV & MwG & MIIS-CV & MIIS/A-CV\\
\cline{2-6}
Mean & 1.087 & 0.002 & 1.000 & 0.011 & 0.002 \\
Variance & 0.805 & 0.001 & 1.000 & 0.011 & 0.001 \\
Covariance & 0.789 & 0.001 & 1.000 & 0.022 & 0.002 \\
$P(X(1)<-2.32)$ & 0.942 & 0.746 & 1.000 & 0.966 & 0.874 \\

 \\
 &  \multicolumn{5}{c}{$\rho=0.5$} \\
 & Gibbs & Gibbs-CV & MwG & MIIS-CV & MIIS/A-CV \\
 \cline{2-6}
Mean & 0.931 & 0.000 & 1.000 & 0.025 & 0.000 \\
Variance & 0.974 & 0.000 & 1.000 & 0.177 & 0.225 \\
Covariance & 0.988 & 0.000 & 1.000 & 0.066 & 0.022 \\
$P(X(1)<-2.32)$ & 0.906 & 0.148 & 1.000 & 0.270 & 0.240 \\

\\
 &  \multicolumn{5}{c}{$\rho=0.25$} \\
 & Gibbs & Gibbs-CV & MwG & MIIS-CV & MIIS/A-CV \\
\cline{2-6}
Mean & 0.944 & 0.000 & 1.000 & 0.073 & 0.000 \\
Variance & 0.830 & 0.000 & 1.000 & 0.493 & 0.850 \\
Covariance & 0.973 & 0.000 & 1.000 & 0.167 & 0.025 \\
$P(X(1)<-2.32)$ & 0.810 & 0.020 & 1.000 & 0.179 & 0.179 \\

\hline\hline
\end{tabular}
\end{threeparttable}
\end{center}
\end{table}

\subsubsection{Mixed Logit Model}

We consider posterior simulation for the Mixed Logit (MIXL) model as a substantive applied example where it is necessary to apply a method such as importance sampling within Gibbs or Metropolis-within-Gibbs. The binary Mixed Logit model specifies the probability that an individual chooses a certain alternative $j=1$ (over $j=0$) at occasion $t$ as
\begin{equation}\label{eq:choiceprob}
p(\textrm{$i$ chooses $j=1$ at $t$}|Z_{it},\beta_{i})=\frac{\exp(\beta_{0i}+\sum_{l=1}^{L}\beta_{li}z_{lit})}{1+\exp(\beta_{0i}+\sum_{l=1}^{L}\beta_{li}z_{lit})},
\end{equation}
where $\delta_{i}=(\beta_{0i},\beta_{1i},\ldots,\beta_{Li})'$ is the vector of utility weights for individual $i$ and $Z_{it}=(z_{1it},\ldots,z_{Lit})'$ is the corresponding vector of attributes for the choice.
The individual specific constants are
$\beta_{0i}=\beta_{0}+\eta_{0i}$ with $\eta_{0i}\sim \N(0,\sigma_{0}^2)$
and the attribute weights for each individual are latent variables with specification
\begin{equation}\label{eq:choiceatt}
\beta_{li}=\beta_l+\eta_{li}, \qquad l=1,\ldots,L,
\end{equation}
with $\eta_{li}\sim \N(0,\sigma_l^2)$.

The parameter vector is $\theta=(\beta_{0},\sigma_{0}^2,\beta_1,\ldots,\beta_L,\sigma_1^2,\ldots,\sigma_L^2)'$, while the vector of latent variables for each individual is $\zeta_{i}=(\beta_{0i},\ldots,\beta_{Li})$.  The Mixed Logit model captures heterogeneity in preferences by allowing individuals to weight the choice attributes differently. By introducing taste heterogeneity, the MIXL specification avoids the restrictive independence of irrelevant alternatives (IIA) property of the standard multinomial logit model \citep{fklw2010}.

We consider an empirical application to the Pap smear data set used for simulated maximum likelihood estimation in \cite{fklw2010}. In this data set, $I=79$ women choose whether or not to have a Pap smear test on $T=32$ choice scenarios. We let the observed choice for individual $i$ at occasion $t$ be $y_{it}=1$ if the woman chooses to take the test and $y_{it}=0$ otherwise. Table \ref{tab:choices} lists the choice attributes and the associated coefficients.  We impose the restriction that $\sigma_5^2=0$ in our illustrations since we have found no evidence of heterogeneity for this attribute. To simplify the computational algorithm for this example given this restriction, we fix $\beta_5$ at the maximum likelihood estimate.

\begin{table}[h]
\caption{Choice attributes for the pap smear data set}\label{tab:choices}
\begin{center}
\begin{tabular}{lll}
\hline\hline
Choice attributes & Values & Associated parameters \\
\hline
Alternative specific constant for test & 1 & $\beta_{0}, \sigma_{0}$ \\
Whether patient knows doctor &  0 (no), 1 (yes) & $\beta_1, \sigma_1$ \\
Whether doctor is male &  0 (no), 1 (yes) & $\beta_2, \sigma_2$ \\
Whether test is due &  0 (no), 1 (yes) & $\beta_3, \sigma_3$ \\
Whether doctor recommends test &  0 (no), 1 (yes) & $\beta_4, \sigma_4$ \\
Test cost & \{0, 10, 20, 30\}/10 & $\beta_5$ \\
\hline\hline
\end{tabular}
\end{center}
\end{table}

We specify the priors as
$\beta_{0}\sim \N(0,100)$, $\sigma_{0} \propto (1+\sigma_{0}^2)^{-1}$,
$\beta_l\sim \N(0,100)$, $\sigma_l \propto (1+\sigma_{l}^2)^{-1}$, for $l=1,\ldots,L$. We follow \cite{Gelman2006} and impose half-Cauchy priors on the standard deviation parameters.

In the general notation of the paper, we want to simulate the posterior distribution of $x = \{\boldsymbol \theta', \zeta_{1}', \ldots, \zeta_{I}'\}'$.
\subsubsection{Results}

We focus on the estimation of the posterior mean of the model parameters, that is
\[E_\pi(\beta_{0}),\,\,E_\pi(\sigma_{0}),\,\,E_\pi(\beta_1),\,\,\ldots,\,\,E_\pi(\beta_4),\,\,E_\pi(\sigma_1),\,\,\ldots,\,\,E_\pi(\sigma_4).\]

We implement MIIS and Metropolis-within-Gibbs algorithms that iteratively sample the parameters ($x(1)=\theta$) and the choice attributes for all individuals ($x(2)=\{\zeta_{1}', \ldots, \zeta_{I}'\}')$ conditional on each other. Equation \eqref{eq:choiceatt} implies that conditional on $\beta_{li}$ for all $i$ and $l=0,1,\ldots,4$, the posterior of $\theta$ factorises into five components with Gaussian conditional likelihoods from which we can independently sample the corresponding mean and standard deviation parameters. As before, the number of importance samples for the MIIS method is $N=50$. We generate 50 iterates of $x(s)$ conditional of the previous value of $x(\sk[s])$ in the MwG algorithm to make the two approaches comparable.  The proposal for the individual choice attributes combines the efficient importance sampling (EIS) method of \cite{ZR2007} with the defensive sampling approach of \cite{Hesterberg1995}. The importance density is the two component defensive mixture
\[q(\zeta_i|y_{i1},\ldots,y_{iT})=\omega q^\text{EIS}(\zeta_i|y_{i1},\ldots,y_{iT})+(1-\omega)p(\zeta_i),\]
where $q^\text{EIS}(x_i|y_{i1},\ldots,y_{iT})$ is a multivariate Gaussian importance density obtained using the EIS method.  Following \cite{Hesterberg1995}, the inclusion of the state prior $p(\zeta_i)$ in the mixture ensures that the importance weights are bounded. We set the mixture weight as $\omega=0.5$. We also use the EIS method to obtain the importance parameters for the five bivariate parameter proposals (the conditional maximum likelihood estimates are easy to implement alternatives which we use to initialise the EIS method) and incorporate antithetic variables throughout.

We consider the same set of twenty control variates for each MIIS estimate. The first set of control variates are based on the parameters $\boldsymbol \theta $,
\[{\wh E}^M_{MC}(\theta_j)-{\wh E}_{MIIS}^{M,N}(\theta_j),\qquad \text{for } j=1,\ldots 10,\]
These control variables are the differences between the standard MCMC posterior mean estimates and the MIIS Rao-Blackwellised estimates. We additionally use two types of control variates based on the individual choice attributes. The first group of control variates based on the attributes is
\[I^{-1}\sum_{i=1}^{I}{\wh E}^M_{MC}(\beta_{ki})-{\wh E}_{MIIS}^{M,N}(\beta_{ki}),\qquad k=0,\ldots,4\]
and the second is
\[I^{-1}\sum_{i=1}^{I}{\wh E}^M_{MC}(\beta_{ki}^2)-{\wh E}_{MIIS}^{M,N}(\beta_{ki}^2),\qquad k=0,\ldots,4.\]
The motivation for this second set of control variates is that the parameters of the model are the means and variances of the individual choice attributes, see equation \eqref{eq:choiceatt}. Since there are $I$ individuals, we construct the control variates by averaging the posterior moment estimates of $\beta_{ki}$.  Because of the correlation between the parameters ($x(1)$) and the choice attributes ($x(2)$) in the Markov chain, we expect these control variates to be highly correlated with the posterior mean estimates of the parameters. Moreover, the use of all twenty control variates simultaneously allows us to leverage the high posterior correlations for variance reduction. We estimate the optimal control variate coefficients as in the last section.

Table \ref{tab:MIXL} reports the estimated MSE for each method relative to MwG. The results are based on 500 independent Markov Chains with 20,000 iterations after 1,000 burn-in draws.  The MIIS column in the table corresponds to the Rao-Blackwellized estimate ${\wh E}_{MIIS}^{M,N}(\theta_j)$ given by \eqref{eq: RB gibbs estimator}. We initialize every chain at the maximum likelihood estimate and approximate the ``true'' posterior means by averaging all the 500 MwG and MIIS estimates (without control variables). The results show that the benefits of using the MIIS Rao-Blackwellized estimates by themselves may be small or negligible because the autocorrelation in the MIIS chain is the main determinant of the total variance of the estimates in this example. When we use the Rao-Blackwellized estimates to construct the control variates, we obtain 75-95\% reductions in MSE relative to the MwG algorithm. The two methods have similar computational cost and implementation effort.

\begin{table}[!htbp]
\begin{center}
\caption{Mixed Logit Application -- Monte Carlo MSE of posterior mean estimates relative to MwG.}\label{tab:MIXL}
\begin{threeparttable}
\begin{tabular}{lccc}
\hline\hline
Parameter & MwG & MIIS & MIIS-CV \\
\hline
$\beta_{0}$ & 1.00 & 0.91 & 0.07 \\
$\beta_1$ & 1.00 & 1.23 & 0.06 \\
$\beta_2$ & 1.00 & 0.92 & 0.05 \\
$\beta_3$ & 1.00 & 0.98 & 0.06 \\
$\beta_4$ & 1.00 & 0.66 & 0.08 \\
$\sigma_0$ & 1.00 & 0.95 & 0.07 \\
$\sigma_1$ & 1.00 & 1.02 & 0.16 \\
$\sigma_2$ & 1.00 & 0.94 & 0.08 \\
$\sigma_3$ & 1.00 & 1.17 & 0.08 \\
$\sigma_4$ & 1.00 & 0.54 & 0.25 \\
\hline\hline
\end{tabular}
\end{threeparttable}
\end{center}
\end{table}
\subsection{Random Walk Importance Sampler}

\subsubsection{Markov Modulated Poisson Process}

A Markov Modulated Poisson Process (MMPP) $Y_t$ is a Poisson process whose intensity $\lambda_t$ takes on a discrete number $d$
of values $\psi=(\psi_1,\ldots,\psi_d)'$, with the intensity at any time point determined by the state of an unobserved continuous-time Markov chain with generator $Q$.  We identify the model by imposing the parameter restriction $\psi_d>\ldots>\psi_1$.  \citet{sfr2010} recently considered the MMPP as a challenging case study for comparing a range of Random Walk Metropolis~(RWM) algorithms proposed in the literature. We replicate their setting to illustrate how the random walk importance sampler of Section \ref{s:randomwalk} can lead to more efficient and robust MCMC simulation compared to standard RW samplers.

Suppose that we observe a realisation of the process over a certain time window and record $n$ event times. \citet{fs2006} derived the likelihood for the model as
\begin{equation}
L(Q,\psi,t)=\nu'\exp^{(Q-\Psi)t_1}\Psi\ldots\exp^{(Q-\Psi)t_n}\Psi\exp^{(Q-\Psi)t_{n+1}}\iota,
\end{equation}
where
\[Q=\left(\begin{matrix} -q_{12} & q_{12} \\ q_{21} & -q_{21} \end{matrix}\right),\]
$\nu$ is the initial distribution of the latent state $Z_t$ (which we take to be the stationary distribution of the chain implied by $Q$), i.e., $\Pr(Z_t = j) = \nu(j)$, $\Psi=\text{diag}(\psi)$, $\iota$ is a vector of ones,  $t_1$ is the time from the start of the observation window until the first event, $t_i$ is the time between events $i-1$ and $i$, and $t_{n+1}$ is time between event $n$ and the end of the observation window.

\subsubsection{Simulation Study}

We replicate the simulation study in \citet{sfr2010}. We simulate the MMPP model with $d=2$ over an observation window of 100 seconds. The generator matrix $Q$ has parameters $q_{12}=q_{21}=1$. The intensity vector is $\psi=(10,17)'$. As in the \citet{sfr2010} application, we complete the model by specifying exponential priors for all the parameters. The means of the priors are the true parameters.

We consider three different methods: the standard RWM algorithm, the multiple-try RWM (MTM) of \citet{llw2000}, and the MIIS random walk method (Algorithm 7).  We consider a random walk on the transformed parameter vector $\widetilde{\theta}=(\log(\psi_1),\log(\psi_2-\psi_1),\log(q_{12}),\log(q_{21})$, which is more efficient than working on the original scale. Let $i$ index the current iteration of the Markov Chain. The proposal for all methods is
\[q_{i+1}(\theta^*_{i+1}|\theta^*_{i})= N\left(\widetilde{\theta}_i,\frac{2.38^2}{4}\widetilde{\Sigma}\right),\]
where $\widetilde{\Sigma}$ is an estimate of the posterior covariance matrix of $\widetilde{\theta}$ based on a trial run of the RWM algorithm. The scale of the proposal aims to achieve an acceptance rate of 0.234 in the standard RWM algorithm, which is optimal rate under certain assumptions; see the discussion in \citet{sfr2010}. We consider four control variates associated with each parameter for estimating each posterior mean
\[{\wh E}^M_{MC}(\theta_j)-{\wh E}_{MIIS}^{M,N}(\theta_j),\qquad \text{for } j=1,\ldots 4.\]
As before, the control variates are the differences between the standard MCMC estimates and the Rao-Blackwellized estimates that reuse all the particles.

We parallelize the likelihood evaluations over eight cores at every iteration of the MIIS Markov chain and set the number of particles to $N=8$ and $N=16$. Our discussion treats the MIIS method with $N=8$ draws as being comparable to the standard RWM method, which is difficult to parallelize. This implies that the MIIS algorithm performs eight times as many likelihood evaluations as the standard RWM algorithm in total, but in the same amount of time under perfect parallelization. We report the actual computing times in the tables. We configure the MTM method such that it performs the same number of likelihood evaluations per iteration as the MIIS algorithm with $N=8$ particles. However, the parallelization for the MTM method is less efficient as every iteration of the method requires two separate stages.

The simulation study averages results over ten independent realisations of the DGP. For every realisation, we simulated 500 independent Markov chains for each method and ran each Markov Chain for 10,000 iterations after discarding a burn-in of 1,000 iterations. We consider two cases for initialisation.  We initialize the algorithm at the maximum likelihood estimate for half the chains.  For the other half, we initialize the chain by drawing from the prior. We use the same draw from the prior to initialize all the methods at each replication. Initializing from the prior allows us to compare the convergence performance of each method. We then compute the posterior mean and variance estimates based on each chain. We combine all chains initialized at the true parameters to obtain precise approximations to the true posterior means and variances.

Table \ref{tab:comparison2} reports the MSE efficiency of the posterior estimates relative to the performance of the RWM method. We average the results over the four parameters for conciseness. We also present the actual computing times, and average acceptance rate and average integrated autocorrelation time (IACT). We base the last two results on longer independent chains with 100,000 iterations (one per simulated dataset). In the case of MIIS, the ``acceptance'' rate is the proportion of iterations in which the sampled particle is not the previous iterate.  We also report the relative time adjusted MSE for each method, which we define as $(\text{time}_{alg} \times \text{mse}_{alg})/(\text{time}_{MH} \times \text{mse}_{MH})$. This estimate approximates the MSE relative to the MH algorithm for the same amount of computing time. We average all these results over realisations.

The results show that the MIIS method reduces the time adjusted MSEs by 70\% compared to MH when we initialize the chain at the true parameters. This gain in performance comes both from reductions in IACT and the use of Rao-Blackwellization to estimate the posterior moments. The MIIS method also outperforms the MTM method, at a lower computational cost. Using control variates further reduces the MSE, leading to a 78-83\ time adjusted improvement over MH. To put these gains in the context of the random walk literature, the best performing algorithm in the simulation of \citet{sfr2010} for the same DGP, a MH random walk in the log scale with with an adaptively tuned mixture proposal, leads to a 84\% reduction in variance (measured by IACT) compared the least efficient algorithm in their analysis, a MH random walk with tuned proposal $N(0, \lambda^2 I)$.  The table also shows that when we initialize all algorithms from the prior, the MIIS algorithm with $N = 8$ particles and control variates generates 80-99\% reductions in time adjusted MSE compared to the standard RWM algorithm. This result suggests that the MIIS algorithm is more robust to the initial conditions than the standard RWM and MTM algorithms.

\begin{table}[!htbp]
\begin{center}
\caption{MMPP DGP with $\psi_1=10$ and $\psi_2=17$. Monte Carlo MSE of posterior estimates relative to MH (average across parameters).}\label{tab:comparison2}
\vspace{3pt}
\footnotesize{We define the time adjusted MSE as $(\text{time}_{alg} \times \text{mse}_{alg})/(\text{time}_{MH} \times \text{mse}_{MH})$, which approximates the MSE relative to the MH algorithm for the same amount of computing time.}
\begin{threeparttable}
\begin{tabular}{lcccccc}
\\
\hline\hline
 &  &  & \multicolumn{2}{c}{MIIS}  &  \multicolumn{2}{c}{MIIS-CV}   \\
 & MH & MTM & N=8 & N=16 & N=8 & N=16  \\
  \cline{2-7}
Acceptance & 0.30 & 0.46 & 0.32 & 0.47 & 0.32 & 0.47 \\
IACT & 15.0 & 9.3 & 7.5 & 5.6 & 7.5 & 5.6 \\
Time & 8 & 15 & 9 & 14 & 9 & 14 \\
 &  &  &  &  &  &  \\
 & \multicolumn{6}{c}{\textbf{Initializing at the true parameters}} \\
 \cline{2-7}
Mean & 1.00 & 0.47 & 0.24 & 0.16 & 0.15 & 0.12 \\
Variance & 1.00 & 0.67 & 0.25 & 0.14 & 0.20 & 0.15 \\
\\
 & \multicolumn{6}{c}{Time adjusted MSE} \\
Mean & 1.00 & 0.88 & 0.27 & 0.29 & 0.17 & 0.21 \\
Variance & 1.00 & 1.26 & 0.29 & 0.25 & 0.22 & 0.27 \\
 &  &  &  &  &  &  \\
 & \multicolumn{6}{c}{\textbf{Initializing from the prior}} \\
 \cline{2-7}
Mean & 1.00 & 0.07 & 0.06 & 0.02 & 0.01 & 0.01 \\
Variance & 1.00 & 0.24 & 0.19 & 0.06 & 0.05 & 0.04 \\
\\
 & \multicolumn{6}{c}{Time adjusted MSE} \\
Mean & 1.00 & 0.13 & 0.07 & 0.03 & 0.01 & 0.01 \\
Variance & 1.00 & 0.46 & 0.21 & 0.10 & 0.06 & 0.06 \\
\hline\hline
\end{tabular}
\end{threeparttable}
\end{center}
\end{table}

\subsubsection{Empirical Example}

We now apply the RWM, MTM and MIIS methods using data from the empirical example in \citet{fs2006}.  The data consists of positions (in bases) of Chi sites (a DNA motif) in the genome of Escherichia coli bacteria. The specification of the MMPP model is the same as in the simulation study above. We follow the procedure described in \citet{fs2006} to obtain data-based parameters for the exponential priors. We estimate the MMPP for the lagging part of the outer ring of the E. coli genome strand, which has 117 observations in total. We ran each Markov Chain for 50,000 iterations after discarding a burn-in of 1,000 iterations and initialized all chains at the maximum likelihood estimate. We also use the Hessian of the likelihood at the maximum likelihood estimate to obtain the shape of the random walk proposal.

Table \ref{tab:comparison5} displays the Monte Carlo MSEs over 500 replications of each algorithm. The results show that the MIIS-CV algorithm with $N=16$ has $83-90\%$ lower MSEs than the standard RWM algorithm. Adjusting for the actual computational times, the improvements are between $44-66\%$. The MIIS algorithm also outperforms the MTM method. We note from the table that the practical computational cost of adding particles tends to be low, so that we can consider a higher $N$ to increase robustness.

\begin{table}[!htbp]
\begin{center}
\caption{Empirical example for the MMPP model -- Monte Carlo MSE of posterior estimates relative to MH.}\label{tab:comparison5}
\vspace{3pt}
\footnotesize{We define the time adjusted MSE as $(\text{time}_{alg} \times \text{mse}_{alg})/(\text{time}_{MH} \times \text{mse}_{MH})$, which approximates the MSE relative to the MH algorithm for the same amount of computing time.}
\begin{threeparttable}
\begin{tabular}{lcccccc}
\\
\hline\hline
 &  &  & \multicolumn{2}{c}{MIIS}  &  \multicolumn{2}{c}{MIIS-CV}    \\
 & MH & MTM & N=8 & N=16 & N=8 & N=16 \\
   \cline{2-7}
Acceptance & 0.24 & 0.54 & 0.42 & 0.57 & 0.42 & 0.57 \\
IACT & 55.2 & 20.2 & 13.3 & 9.4 & 13.3 & 9.4 \\
Time & 20 & 62 & 55 & 59 & 58 & 63 \\
\\
&  \multicolumn{6}{c}{\textbf{Posterior mean}} \\
   \cline{2-7}
$\psi_1$ & 1.00 & 0.45 & 0.25 & 0.17 & 0.18 & 0.13 \\
$\psi_2$ & 1.00 & 0.50 & 0.21 & 0.15 & 0.15 & 0.10 \\
$q_{12}$ & 1.00 & 0.39 & 0.20 & 0.14 & 0.16 & 0.10 \\
$q_{21}$ & 1.00 & 0.52 & 0.27 & 0.15 & 0.15 & 0.12 \\
\\
&  \multicolumn{6}{c}{Time adjusted MSE}\\
$\psi_1$ & 1.00 & 1.39 & 0.68 & 0.50 & 0.51 & 0.43 \\
$\psi_2$ & 1.00 & 1.55 & 0.59 & 0.45 & 0.45 & 0.33 \\
$q_{12}$ & 1.00 & 1.21 & 0.54 & 0.42 & 0.46 & 0.33 \\
$q_{21}$ & 1.00 & 1.60 & 0.73 & 0.45 & 0.44 & 0.37 \\
\\
&  \multicolumn{6}{c}{\textbf{Posterior variance}} \\
   \cline{2-7}
$\psi_1$ & 1.00 & 0.53 & 0.34 & 0.19 & 0.22 & 0.17 \\
$\psi_2$ & 1.00 & 0.65 & 0.29 & 0.21 & 0.18 & 0.11 \\
$q_{12}$ & 1.00 & 0.45 & 0.20 & 0.12 & 0.17 & 0.12 \\
$q_{21}$ & 1.00 & 0.50 & 0.23 & 0.14 & 0.15 & 0.10 \\
\\
&  \multicolumn{6}{c}{Time adjusted MSE}\\
$\psi_1$ & 1.00 & 1.66 & 0.93 & 0.56 & 0.64 & 0.56 \\
$\psi_2$ & 1.00 & 2.02 & 0.81 & 0.63 & 0.51 & 0.36 \\
$q_{12}$ & 1.00 & 1.39 & 0.55 & 0.35 & 0.50 & 0.38 \\
$q_{21}$ & 1.00 & 1.55 & 0.63 & 0.41 & 0.45 & 0.33 \\
\hline\hline
\end{tabular}
\end{threeparttable}
\end{center}
\end{table}

\section{Theory} \label{theory}
This section presents our theoretical results for the MIIS estimators and restates some of the definitions in previous sections in a more general setting.

Let $(A,\Omega)$ denote a measurable space and $\pi$ some given \textit{target} probability distribution on $(A,\Omega)$. Assume that a reference measure $\mu$ dominates $\pi$ ($\pi\ll\mu$) and that $\pi(\dx) = \pi(x)\mu(\dx)$. With a small abuse of notation, we write $\pi(x)$ for the density of the probability measure $\pi$ with respect to $\mu$.  In most situations $A\subseteq\Re^d$ ($d\in\mathbb{N}$, the set of positive integers), $\Omega = \mathcal{B}(A)$ is the Borel $\sigma$-algebra of the set $A$, and the majorizing measure $\mu$ is either the counting measure, the Lebesgue measure, or a combination of both. We assume that the distributions $\eta(\dxi|y)$ and $q_k(\dx|\xi)$ ($k=1,\dots,N$) admit densities $\eta(\xi|y)$ and $q_k(x|\xi)$ with respect to the same measure $\mu$. We work interchangeably with other distributions and their corresponding densities. Let $A = \times_{i=1}^dA_i$ and $\Omega = \Omega_1\otimes\dots\otimes\Omega_d$. In the conditional case, for all $s=1,\dots,d$, let $(A_s,\Omega_s)$ denote a measurable spaces. We assume that $\pi_s(\dd{y}(s)|y(\sk[s]))$, $q_{s,k}(\dd{y}(s)|\xi_s,y(\sk[s])$, and $\eta_s(\dxi_s|y(s),y(\sk[s]))$ are defined on $(A_s,\Omega_s)$ and have densities with respect to some majorizing measure $\mu_s$, that may depend on $y(\sk[s])\in A_{\sk[s]}$ for $A_{\sk[s]}= \times_{i\ne s}^dA_i$.


\subsection{Convergence of the marginal MIIS chain}\label{SS: convergence of marginal}
If $(y,k)$ is marginally distributed as $N^{-1}\pi$, then the CIS estimator is unbiased by Theorem,~\ref{thm:CISunbiased}(iii). Theorem~\ref{thm:convergence} below shows that the MIIS Algorithm (Algorithm \ref{alg: MIIS algorithm} in Section~\ref{s:miisfull}) samples from the target density $N^{-1}\pi$ asymptotically, i.e., as the number of iterations $t\rightarrow \infty$. In other words, the marginal distribution of $(y^{(t)},k^{(t)})$ is $N^{-1}\pi$, asymptotically.

For all $l,k\in\{1{:}N\}$ and $y,z\in A$, define
\begin{align} \label{eq: def S_kl}
\begin{split}
S_{l,k}(\xi,y,z):= & \int_{A^2}T(y,\dx_l;\xi)T(z,\dx_k;\xi)\frac{\bs{q}_{l,k}(x_l,x_k|\xi)}{q_l(x_l|\xi)q_k(x_k|\xi)}, \quad   l \neq k \\
S_{l,l}(\xi,y,z):= & \int_{A}\frac{T(y,\dx_l;\xi)T(z,x_l;\xi)}{q_l(x_l|\xi)}, \quad l = k
\end{split}
\end{align}
where $\bs{q}_{l,k}(x_l,x_k|\xi)$ is the joint marginal of $(x_l,x_k)$ for $l \neq k $.

The proof of Theorem~\ref{thm:convergence} is based on the following assumption discussed in Section~\ref{SS: Examples theory}.

\begin{assumption} \label{a:regcis}
\begin{enumerate}
    \item [(i)]  There exists a constant $C$, $0<C<\infty$,  such that the marginal densities $q_{k}(x|\xi)$ satisfy $\pi(x)\eta(\xi|x)\le C\,q_{k}(x|\xi)$, for each $k$ and all $x,\xi\in A$.
    \item [(ii)]
    \begin{itemize}
       \item [(a)] For each $k,l \in  \{1{:}N\}$ and $y,z \in A$,  there exist functions $h_{k,l}(y,z) $ such that
       \[
         \int_A S_{l,k}(\xi,y,z)\,\eta(\xi|y)\eta(\xi|z)\mu(\dxi)  \geq  h_{l,k}(y,z).
       \]
    \item [(b)] For each $l \in  \{1{:}N\}$ there exists a set $\mathcal{J}_l\subseteq\{1{:}N\}\setminus\{l\}$ such that: $\mathcal{J}_l\cap\mathcal{J}_k\ne \emptyset$ for $l\ne k$; and
    \item [(c)] for all $j\in\mathcal{J}_l$ and $y,z \in A$  $h_{l,j}(y,z)>0$ and $h_{j,l}(y,z)>0$.
    \end{itemize}
\end{enumerate}
\end{assumption}
Assumption \ref{a:regcis}~(i) requires the weights to be uniformly bounded and it is often used in the particle literature. This condition is not restrictive and can be enforced by choosing suitable $q_k$ and $\eta$. Part (ii) is a technical condition that imposes regularity conditions on the pairwise dependence of the particles, on the kernel $T$, and the auxiliary distribution $\eta$.

\begin{theorem} \label{thm:convergence}

\begin{itemize}
\item [(i)] If  Assumption~\ref{a:regcis} holds then the marginal chain $\{(y^{(t)},k^{(t)})\}$, sampled using MIIS,  is Markov and ergodic, i.e.,
\begin{align*}
\lim_{t \rightarrow \infty} \mid P^t(l,y;\cdot) - N^{-1}\pi(\cdot) |_{TV} & = 0,
\end{align*}
where $P(y,l;B\times\{k\})$ is the Markov transition kernel from $(y,l)$ to $B\times\{k\}$, $B\in\Omega$, $k\in\{1{:}N\}$.
 \item [(ii)]
If, in addition, for $ k \in \mathcal{J}_l$, $h_{l,k}(y,z) \geq  \underbar{h}_{l,k}(z)> 0$, i.e., $\underbar{h}_{l,} $  does not depend on the initial value $y\in A$, then the marginal chain is uniformly ergodic.
\end{itemize}
\end{theorem}
The distribution of the marginal chain $\{(y^{(t)},k^{(t)})\}$ converges to the target distribution $N^{-1}\pi$ as the number of iterations increases. It means that, after a warm up period, the marginal distribution of samples from the chain is $N^{-1}\pi$ and, hence, ${\wh E}^M_{MC}(f)$ is an unbiased estimator of $E_\pi(f)$ for any integrable $f$.  If $E_\pi(|f|)< \infty$, then by  Theorem 3 of \cite{tierney1994}, ergodicity implies that ${\wh E}^M_{MC}(f)$ is also a consistent estimator of $E_\pi(f)$.  If $E_\pi(f^2)< \infty$ and uniform ergodicity holds then by Theorem ~5 of \cite{tierney1994} we also obtain a central limit theorem for ${\wh E}^M_{MC}(f)$.

\subsection{Convergence results for the examples in Section~\ref{examples}}\label{SS: Examples theory}
This section discusses the application of Theorem~\ref{thm:convergence} to the examples in Section~\ref{examples}. In all three examples $T$ is the identity kernel, i.e., $T(y,dz; \xi) = \delta_y(dz)$. This gives $S_{l,k}(\xi,y,z) = \bs{q}_{l,k}(y,z|\xi)/q_l(y|\xi)q_k(z|\xi)$ for $ k \neq l $ and $S_{l,l}(\xi,y,z)  = I(z=y)/q_l(y|\xi)$. Hence, we require $h_{l,k}(y,z)\ge 0$ functions such that
\begin{align*}
\int_A\eta(\xi|y)\eta(\xi|z)\frac{\bs{q}_{l,k}(y,z|\xi)}{q_l(y|\xi)q_k(z|\xi)}\mu(d\xi) & \geq h_{l,k}(y,z) , \quad \text{for} \quad l \neq k\\
I(z=y) \int_A \frac{\eta(\xi|y)\eta(\xi|z)}{q_l(y|\xi) } \mu(\dxi) \geq h_{l,l}(y,z) & .
\end{align*}
Part (i) is assumed explicitly and we choose  $\mathcal{J}_l$ to satisfy Assumption \ref{a:regcis}(ii).
\noindent
\subsubsection*{Simple importance sampling example} This example is discussed in Section~\ref{SS: simple IS}.

\begin{corollary}\label{C: miis sis}
Suppose that there is no dependence of $T$ and the $q_i$ on $\xi$ and (i) $T(x,\dd{y}) = \delta_x(\dd{y})$, (ii) $q(\dx_{1:N}) = \prod_{i=1}^N q_i(\dx_i)$ and (iii) $\pi(\dx_i)\le C\,q_i(\dx_i)$, where $C > 0 $  is a positive constant.  Then, the marginal chain $\{(y^{(t)},k^{(t)})\}$ is uniformly ergodic for $N \geq 3$.
\end{corollary}

Let $\eta(\dxi|y) = \delta_0(\dxi)$, without loss of generality. It is easy to see that $h_{l,k}(y,z) = I(l\ne k)$ is a valid choice. Assumption \ref{a:regcis} is satisfied by taking $\mathcal{J}_l = \{1{:}N\}\setminus \{l\}$, and $N\ge 3$. Uniform ergodicity follows from Theorem \ref{thm:convergence} part~(ii), because $h_{k,l}(y,\cdot)$ does not depend on $y$ for $l\in\mathcal{J}_k$.

\noindent
\subsubsection*{Importance sampling with antithetic variables} This example is discussed in Section~\ref{s:anti}.
\begin{corollary}  \label{l:ISA-CIS}
Suppose there is no dependence on $\xi$ and (i) $T(x,\dd{y}) = \delta_x(\dd{y}) = \delta_x(y)\dd{y} $, (ii) $q(\dx_{1:N}) = \prod_{i=1}^{N/2} q_{i,i+N/2}(\dx_i,\dx_{i+N/2})$ such that $q_{i+N/2}(\dx_{i+N/2}|x_i) = \delta_{Q_i^{-1}(1 - Q_i(x_i) )} (\dx_{i+N/2}) $, where   $Q_i$ is the cdf of $q_i(x_i)$. (iii) $\pi(\dx_i) \le C\,q_i(\dx_i)$, where $C > 0 $ is a positive constant  Then, the marginal chain $\{(y^{(t)},k^{(t)})\}$ is uniformly ergodic for $N/2 \geq 3$.
\end{corollary}

It is straightforward to check that,
 \[
\frac{\bs{q}_{l,k}(y,z)}{q_l(y)q_k(z)} =
 \begin{cases}
 \frac{\delta_{Q^{-1}(1 - Q(y))}(z)}{q(z)} & k\in\{1{:}N\}\cap\{l-N/2,N/2+l\},\,k\ne l\\
 1 & k\in\{1{:}N\}\setminus\{l,l-N/2,N/2+l\}
 \end{cases}.
 \]
 Choose $\mathcal{J}_l = \{1{:}N\}\setminus\{l,l-N/2,N/2+l\}$. It is easy to see that $h_{l,k} = I(k\in\mathcal{J}_l)$ is a valid choice. Assumption \ref{a:regcis} is satisfied and the MIIS sampler is uniformly ergodic using the same arguments as in the previous example.

\subsubsection*{Random walk importance sampler} This example is discussed in  Section~\ref{s:randomwalk}.
\begin{corollary} \label{l:RWIS-CIS}
Suppose that (i) $T(x,\dd{y}|\xi) = \delta_x(\dd{y}) = \delta_x(y)\dd{y} $, (ii) $\bs{q}(\dx_{1:N}|\xi) = \prod_{i=1}^N q_i(\dx_i| \xi)$ and (iii) $q_i(\dx_i| \xi) = \phi(x_i - \xi) \dx_i $; (iv) $\eta(\dxi|y) = \phi(\xi-y)\dxi$; (v) $\phi(x-y) > 0 $ for any $x,y \in A$. (vi) $ \pi(x_i)\leq C$. Then, the marginal chain $\{(y^{(t)},k^{(t)})\}$ is ergodic for $N \geq 3$. If $\inf_{z,y \in A}  \int \phi(\xi - z) \phi(\xi - y)\dxi > \varepsilon>0$ Then $\{(y^{(t)},k^{(t)})\}$ is uniformly ergodic.
\end{corollary}
For $l \neq k$, $S_{l,k}(\xi,y,z) = 1$  because the proposals are independent, and
\[
h(y,z):=h_{k,l}(y,z) = \int_A\phi(\xi-y)\phi(z-\xi)\dxi > 0.
\]
Choose $\mathcal{J}_l = \{1{:}N\}\setminus\{l\}$. Then Assumption \ref{a:regcis} holds and ergodicity follows from Theorem \ref{thm:convergence}. By assumption there exists $\epsilon>0$ such that $h(y,z) \geq \epsilon$ for all $y,z \in A$. By defining $ \underbar{h}_{l,k}(z)= \epsilon$ for $k \in \mathcal{J}_l $, uniform ergodicity follows from part (ii) of Theorem \ref{thm:convergence}.

\subsection{The MIIS Gibbs Sampler}
This section shows that the marginal chain $\{l_{1:d}^{(t)}, Y^{(t)}\}$ generated by the MIIS Gibbs sampler (Algorithm~\ref{alg:gibbs miis} in Section~\ref{SS: Gibbs MIIS algorithm}) is ergodic if (i)~the \textit{ideal Gibbs sampler}, i.e., the Gibbs sampler drawing variables from the conditionals $\pi_s(\dd{y}(s)|y(\sk[s]))$, is irreducible and aperiodic;  (ii) the CIS  Gibbs sampler satisfies regularity conditions that are similar to Assumption \ref{a:regcis}, but hold for each $s=1,\dots,d$; (iii) The space $A$ is Euclidean with Lebesgue measure the underlying measure.

Our notation  assumes that we condition on $y(\sk[s])$ when dealing with the $s$th component and do not usually show this conditioning explicitly. The transition kernel for the ideal Gibbs sampler is
\begin{align} \label{eq: ideal gs}
P_G(y;\dd{z} ): = \prod_{i=1}^m \pi_s(\dd{z}(s)|z(1\dcolon s-1),y(s+1\dcolon d))
\end{align}
For all $l,k\in\{1{:}N\}$ and $\xi(s), y(s),z(s)\in A_s$, define
\begin{align} \label{eq: def S_kl gibbs}
\begin{split}
S_{s,l,k}(\xi(s),y(s),z(s)):= & \int_{A^2}T_s(y(s),\dx_l(s);\xi(s))T_s(z(s),\dx_k(s);\xi(s)) \\ & \times \frac{\bs{q}_{s,l,k}(x_l(s),x_k(s)|\xi(s))}{q_{s,l}(x_l(s)|\xi(s))q_{s,k}(x_k(s)|\xi(s))}, \quad   l \neq k \\
S_{s,l,l}(\xi(s),y(s),z(s)):= & \int_{A}\frac{T_s(y(s),\dx_l(s);\xi(s))T_s(z(s),x_l(s);\xi(s))}{q_{s,l}(x_l(s)|\xi(s))}, \quad l = k
\end{split}
\end{align}
where $\bs{q}_{s,l,k}(x_l(s),x_k(s)|\xi(s))$ is the joint marginal of $(x_l(s),x_k(s))$ for $l \neq k $.

The proof of Theorem~\ref{thm:convergencecond} is based on the following assumption, which generalizes Assumption \ref{a:regcis} to the Gibbs case.
\begin{assumption} \label{a:regciscond}
The following condition holds for all $s=1,\ldots,d$. All terms are conditional on $y(\sk[s])$, unless stated otherwise.
\begin{enumerate}
   \item [(i)]  There exists a constant $C$, $0<C<\infty$,  such that the marginal densities $q_{s,k}(x(s)|\xi(s))$ satisfy $\pi_s(x(s))\eta(\xi(s)|x(s))\le C^{1/d}\,q_{s,k}(x(s)|\xi(s))$, for each $k$ and all $x(s),\xi(s)\in A_s$.
    \item [(ii)] For each $k,l \in  \{1{:}N\}, y(s),z(s) \in A_s$ and $y(\sk[s]) \in A_{\sk[s]}$,
    \begin{itemize}
    \item [(a)] There exist functions $h_{s,k,l}(y(s),z(s))\ge 0 $ such that
    \[
      \int_{A_s} S_{s,l,k}(\xi(s),y(s),z(s))\,\eta(\xi(s)|y(s))\eta(\xi(s)|z(s))\mu(\dxi(s))  \geq  h_{s,l,k}(y(s),z(s));
    \]

   \item [(b)] for each $l \in  \{1{:}N\}$, there exists a set $\mathcal{J}_{s,l}\subseteq\{1{:}N\}\setminus \{l\}$  $\mathcal{J}_{s,l}\cap\mathcal{J}_{s,k}\ne \emptyset$ for $l\ne k$;
   \item [(c)] for each $j\in\mathcal{J}_{s,l}$, $h_{s,l,j}(z(s),y(s))>0$ and $h_{s,j,l}(z(s),y(s))>0$ on $y\in\{x\in A:\pi(x)>0\}$ and $z(s)\in A_s$.
    \end{itemize}
\end{enumerate}
\end{assumption}

Define $\bs{l}:=l_{1:d}$ and $ \bs{k}:= k_{1:d}$.
\begin{theorem}\label{thm:convergencecond}
	Suppose Assumption \ref{a:regciscond} holds. If $P_G$ is irreducible and aperiodic, then so is the marginal kernel $P_M(y,\bs{l} ; \dd{z}\times \bs{k})$, and for any starting values $y\in A$ with $\pi(y)>0$ and $\bs{l}\in\{1{:}N\}^d$,
			\[
				\lim_{t\rightarrow\infty}\left|P_{M}^t(y,\bs{l};\cdot)-N^{-d}\pi(\cdot)\right|_{TV} = 0.
			\]
\end{theorem}

\subsubsection*{Gibbs Sampler with simple importance sampling example}\label{ex: gibbs with IS}
We now consider the example of the Gibbs sampler with simple importance sampling discussed in Section ~\ref{s:gibbs}.
\begin{corollary}
   Suppose that there is no dependence on $\xi$ and the following conditions hold for $s=1, \dots, d$. (i)~$T_s(x(s),\dd y(s)|y(\sk[s])) = \delta_{x(s)}(\dd y(s))$, (ii)~$q_s(\dd y_{1:N}(s)|y(\sk[s])) = \prod_{i=1}^Nq_{s,i}(\dd y_i(s)|y(\sk[s]))$, (iii) There is a $C > 0 $ such that $q_{s,i}(\dd y_i(s)|y(\sk[s]) ) \geq C^{1/d} \pi_s(\dd y_i(s)|y(\sk[s]) )$. If we further assume that the ideal Gibbs sampler, $P_G$, is irreducible and aperiodic, then the distribution of the marginal chain $\{\bs{l}^{(t)}, y^{(t)}, t \geq 1 \}$ converges to the full target $N^{-d}\pi(\cdot)$ as $t\rightarrow\infty$ for any fixed $N\ge 3$.
\label{l:IS-CIScond}
\end{corollary}

This corollary follows after the same arguments used in the marginal case. The functions $h_{s,l,k} = I(l\ne k)$ and the sets $\mathcal{J}_{s,l} = \{1{:}N\}\setminus \{l\}$ for each $s=1,\dots,d$. The result follows from Theorem \ref{thm:CISunbiased-conditional}.

\subsection{Consistent estimation of expectations}\label{s:VR theory}
\subsubsection*{Using all the particles}\label{SS: using all the particles theory}
The next theorem shows that the MIIS estimator ${\wh E}^{M,N}_{MIIS}(f)$ discussed in Section \ref{s:VR} converge to $E_\pi(f)$.
\begin{corollary}
Let $f:A\mapsto\mathbb{R}$ be such that $E_\pi(|f|)<\infty$ and suppose~Assumption~\ref{a:regcis} holds. Then the
 MIIS estimator ${\wh E}^{M,N}_{MIIS}(f)\rightarrow E_\pi(f)$ with probability one as $M\rightarrow\infty$, for any $N\ge 2$.
\label{thm:MIISconsistent}
\end{corollary}

\subsubsection*{Using Rao Blackwellized estimators}\label{s: rao black theory}
Define the Rao-Blackwellized estimators ${\wh E}^{M,N}_{s,RB}(f)$ and ${\wh E}^{M,N}_{RB}(f)$ as in Section xxx. Then,
\begin{corollary}
Let $f:A\mapsto\mathbb{R}$ be such that $E_\pi(|f|)<\infty$. Suppose Assumption \ref{a:regciscond} holds.
 Then, the Rao-Blackwellized estimators ${\wh E}^{M,N}_{s,RB}(f)$ and ${\wh E}^{M,N}_{RB}(f)$ converge to $E_\pi(f)$ with probability 1 as $M \rightarrow \infty$.
\label{thm:RBconsistent}
\end{corollary}

\subsubsection*{Using Control Variates}\label{SS: using control variates theyr}
The following two results shows that the estimators based on control variates discussed in Section \ref{s:CV}
are  consistent under ergodicity.
\begin{corollary}
Let $f:A\mapsto\mathbb{R}$ be such that $E_\pi(|f|)<\infty$. Suppose~Assumption \ref{a:regcis} holds.
 Then the estimator using control variates ${\wh E}^{M,N}_{CV}(f,\bs\theta)\rightarrow E_\pi(f)$ with probability one as $M\rightarrow\infty$, for  any $\boldsymbol\kappa\in\mathbb{R}^p$.
\label{thm:CVconsistent}
\end{corollary}

\begin{corollary}
	For any $s=1,\dots,d$, let $f:A\mapsto\mathbb{R}$ be such that $E_\pi(|f|)<\infty$. Suppose Assumption \ref{a:regciscond} holds. Then
  the estimator using control variates ${\wh E}^{M,N}_{s,CV}(f,\bs\theta)\rightarrow E_\pi(f)$ with probability one as $M\rightarrow\infty$ and any $\boldsymbol\kappa\in\mathbb{R}^{p_1+\dots+p_d}$.
\label{thm:CVconsistentRB}
\end{corollary}
\bibliographystyle{ecta}
\bibliography{pfmcmc,references}

\appendix
\section{Proofs}
The  notation in this section is the same as in Section~\ref{theory}.
\subsection{Markov Interacting Importance Sampler} 
\begin{proof}[Proof of Theorem~\ref{thm:CISunbiased}]
Part (i): From \eqref{eq:fullmiis}, ${\wt \pi}^N$ is a proper distribution function that integrates to 1 and
has marginal ${\wt \pi}^N(dy) = \pi(dy)$. Part~(ii):
The joint distribution ${\wt \pi}^N (\dx_{1:N},\dxi,k)$ is
	\begin{align*}
		{\wt \pi}^N(\dx_{1:N},\dxi,k) &=  \int_A\pin(\dx_{1:N},\dxi,\dy,k)\\
		&= \int_AN^{-1} \pi(\dy)\eta(\dxi|x)T(y,\dx_k;\xi)\,\bs{q}_{\sk}(\dx_{\sk}|x_k,\xi)	\\
		&= \int_AN^{-1} \pi(\dx_k) \eta(\dxi|x_k)T(x_k,\dy;\xi)\,\bs{q}_{\sk}(\dx_{\sk}|x_k,\xi)	\\
		&= N^{-1}\pi(\dx_k)\eta(\dxi|x_k)\,\bs{q}_{\sk}(\dx_{\sk}|x_k,\xi)\\
		&= \frac{\pi(x_k)\eta(\dxi|x_k)}{Nq_k(x_k|\xi)}\,\bs{q}(\dx_{1:N}|x_k,\xi)\\
		&= \frac{w_k(x_k|\xi)}{N\int_Am(x)\mu(\dx)}{\bf q}(\dx_{1:N}|\xi),
	\end{align*}
	The second line is the joint distribution, the third line follows from reversibility of the Markov kernel, the fourth line integrates out $y$, and the last line follows from the definition of the weights. The conditional distribution
	\begin{align*}
		{\wt \pi}^N(K=k | x_{1:N},\xi) &= \frac{\pin(x_{1:N},\xi,k)}{\sum_{i=1}^N\pin(x_{1:N},\xi,i)}
		=\frac{{w_k(x_k|\xi)}}{\sum_{i=1}^N w_i(x_i|\xi)}
		= W_k(x_{1:N},\xi),
	\end{align*}
Similarly,
	\begin{align*}
\pin(\dy|x_{1:N},\xi,k) &=  \frac{\pin(\dx_{1:N},\dxi,\dy,k)}{\pin(\dx_{1:N},\dxi,k)}\\
		&=\frac{N^{-1}\pi(\dy)\eta(\dxi|x)T(x,\dx_k;\xi)\,\bs{q}_{\sk}(\dx_{\sk}|k,\xi)}{N^{-1}\pi(\dx_k)\eta(\dxi|x_k)\,\bs{q}_{\sk}(\dx_{\sk}|x_k,\xi)} \\
		&= T(x_k,\dy;\xi)\frac{\pi(\dx_k)\eta(\dxi|x_k)\,\bs{q}_{\sk}(\dx_{\sk}|x_k,\xi)}{\pi(\dx_k)\eta(\dxi|x_k)\,\bs{q}_{\sk}(\dx_{\sk}|x_k,\xi)} \\
		&= T(x_k,\dy;\xi).
	\end{align*}

Part (iii):
\begin{align*}
{\wt \pi}^N (k,\dx_k) & = \int {\wt \pi}^N (k, \dy,\dxi,\dx_{1:N}) = N^{-1}\int \pi(\dy)\eta(\dxi|y)T(y,\dx_k; \xi)\bs{q}_{\sk}(\dx_{\sk}|x_k,\xi)\\
& = N^{-1} \pi(\dx_k)\int \eta(\dxi|x_k)T(x_k,\dy; \xi)\bs{q}_{\sk}(\dx_{\sk}|x_k,\xi)  = N^{-1} \pi(\dx_k)\, .
\end{align*}
Hence,
\begin{align*}
E_{{\wt \pi}^N}(f(X_K))&  = \sum_{k=1}^N N^{-1}\int \pi(\dx_k) f(x_k) = E_\pi(f).
\end{align*}

Similarly, by first conditioning on $X_{1:N}$ and $\xi$, we obtain
\begin{align*}
E_{{\wt \pi}^N}(f(X_K))& =E_{{\wt \pi}^N} \Big ( E_{{\wt \pi}^N (\cdot |x_{1:N},\xi)}f(X_K) \Big ) \\
& = E_{{\wt \pi}^N} \Big (N^{-1}\sum_{k=1}^N \int   f(x_k)W_k(x_{1:N},\xi)  T(x_k,\dy;\xi) \Big) \\
& = E_{{\wt \pi}^N} \Big ({\wh E}^N_{CIS} (f) \Big)
\end{align*}
\end{proof}

\begin{proof} [Proof of Theorem~\ref{thm:fullmiis}.]
 The proof follows from Part~(ii) of Theorem~\ref{thm:CISunbiased} and
 because\\ ${\wt \pi}^N(\dxi, dx_{1:N}|y,k) = \Gamma^N(\dxi, dx_{1:N}|y,k) $
\end{proof}

\subsection{Markov Interacting Importance Sampler for Conditional Distributions} 
\begin{proof}[Proof of Theorem~\ref{thm:CISunbiased-conditional}]
	The proof is analogous to the proof of Theorem \ref{thm:CISunbiased}, with $\pi$ replaced by $\pi_s(\cdot|x(\sk[s]))$.
\end{proof}

\begin{proof}[Proof of Theorem \ref{thm:conditionalmiis}]
	We write the MIIS Gibbs sampler as a Gibbs sampler in an augmented space. Each step of the algorithm consists in sampling from the following collapsed Gibbs sampler.
\begin{algorithm} \label{alg: alg gibbs theory}
For $s=1,\dots,d$,
	\begin{enumerate}
		\item [(i)]Sample $X_{1:N}(s),\xi(s)|(y(s), k_s, y(\sk[s]),\xi(\sk[s]),(x_{1:N}(\sk[s]) ), k_{\sk[s])}$ \\from $\Gamma_s^N(\dx_{1:N}(s),\dxi(s)|y(s),k_s)$; and
		\item [(ii)]Sample $Y(s),K_s|x_{1:N}(s),\xi(s), (y(\sk[s]), \xi(\sk[s]))$ from
	\[	
				\sum_{i=1}^NW_{s,i}(x_{1:N}(s);\xi(s)))I(K_s=i)T(x_{i}(s),\dy(s);\xi(s)).
	\]
	\end{enumerate}
\end{algorithm}

To prove the theorem it is sufficient to show that the conditional density
$${\wt \pi}^N (\dy(s), \dxi(s), k_s, \dx_{1:N}(s)| y(\sk[s]), \xi(\sk[s]),k_{\sk[s]},x_{1:N}(\sk[s])) $$ gives the $s$th step in
Algorithm~\ref{alg: alg gibbs theory} above. The proof uses the same arguments as those in Theorem~\ref{thm:fullmiis}.
The joint distribution
\[
		\pin(\dy,\dxi(\sk[s]),(\dx_{1:N}(i), i\in\sk[s]), k_{1:d}) = \frac{\pi(\dy)}{N^d}\prod_{i\ne s}^d\Gamma_i^N(\dx_{1:N}(i),\dxi(i)|x(i),k_i,y(\sk[i])),
	\] after integrating out $(x_{1:N}(s),\xi(s))$. Hence,
the conditional joint distribution	
	\[
		\pin(\dx_{1:N}(s),\dxi(s)|\dy,\xi(\sk[s]),(x_{1:N}(i), i\in\sk[s]), k_{1:d}) =
\Gamma_s^N(\dx_{1:N}(s),\dxi(s)|x(s),k_s,x(\sk[s])),
	\]
	which is consistent with part (i) of Algorithm~\ref{alg: alg gibbs theory}.
	Similarly, $\pin(\dy,\dxi(s),\dx_{1:N}(s),k_s) = N^{-1}\pi(\dy)\times \Gamma_s^N(\dx_{1:N}(s),\dxi(s)|y(s),k_s,y(\sk[s]))$, so
	\begin{align*}
		\pin(\dy(\sk[s]),\dxi(s),\dx_{1:N}(s),k_s) &=  \int_{A_s}N^{-1}\pi(\dy(s),\dy(\sk[s])) \times \Gamma_s^N(\dx_{1:N}(s),\dxi(s)|x(s),k_s,y(\sk[s]))\\
		&= N^{-1}\pi_{\sk[s]}(\dy(\sk[s]))\bs{q}_{s}(\dx_{1:N}(s)|\xi(s),\dy(\sk[s]))\\
		&\times\int_{A_s} \frac{\eta_s(\dxi(s)|y(s),y(\sk[s])) \pi_s(\dy(s)|y(\sk[s]))} {q_{s,k_s}(\dx_{k_s}(s)|\xi(s),y(\sk[s]))} T(y(s),\dx_{k_s}(s);\xi(s),y(\sk[s]))\\
		&= N^{-1}\pi_{\sk[s]}(\dy(\sk[s]))\bs{q}_{s}(\dx_{1:N}(s)|\xi(s),y(\sk[s]))\\
&\frac{\eta_s(d\xi(s)|x_{k_s}(s),y(\sk[s]))
\pi_s(x_{k_s}(s)|\x(\sk[s]))}{q_{s,k_s}(x_{k_s}(s)|\xi(s),x(\sk[s]))}\\
		&\propto \frac{\pi_{\sk[s]}(\dy ( \sk[s]))}{N}\bs{q}_{s}(x_{1:N}(s)|\xi(s),x( \sk[s]))\, w_{s,k_s}(x_{k_s}(s),\xi(s),x(\sk[s])).
	\end{align*}
	Hence, $\Pr(K_s=k_s|y(\sk[s]),\xi(s),x_{1:N}(s)) =W_{s,k_s}(x_{1:N}(s),\xi(s),y(\sk[s]))$, which is consistent with Step~(ii) of Algorithm~\ref{alg: alg gibbs theory}
Following the same arguments as in the proof of Theorem \ref{thm:fullmiis}, we can check that
$\pin(\dy(s)|y( \sk[s]),\xi(s),x_{1:N}(s),k_s)  = T(x_{k_s}(s), \dy(s), \xi(s),x( \sk[s]))$.
	Finally, one can verify that the algorithm targets $\pi$ by first integrating out $(x_{1:N}(i), \xi(i))$, $i=1,\dots,d$, and then summing over $k_1,\dots,k_d$.
	\end{proof}

\subsection{Convergence of MIIS}
Before proving  Theorem \ref{thm:convergence}, we obtain a preliminary lemma.

\begin{lemma} \label{lemma: prel to thm 5}
Suppose Assumption~ \ref{a:regcis} holds. Then,
\begin{enumerate}
\item [(i)]
\begin{align*}
 P(y,l;\dd{z}\times\{k\}) &\ge \frac{1}{C}\frac{\pi(\dd{z})}{N} h_{l,k}(y,z).
 \end{align*}
\item [(ii)]
Recursively define  $H_{l,k}(y,z) = h_{l,k}(y,z)$ and
\begin{align*}
H^{t+1}_{l,k}(y,z) := \E_{N^{-1}\pi}[H^t_{l,J}(y,V)h_{J,k}(V,z)] = \sum_{j=1}^N  N^{-1} \int_A H^t_{l,j}(y,v) h_{j,k}(v,z) \pi(dv).
\end{align*}
Then,
 \begin{equation}
        P^t(y,l;\dd{z}\times\{k\})\ge \left(\frac{1}{C}\right)^t\,\frac{\pi(\dd{z})}{N}\,H^t_{l,k}(y,z).
        \label{eq:lbndP}
    \end{equation}
\item [(iii)]
$H^t_{l,k}(y,z) > 0 $ for $ t \geq 2$ for all $y,z \in A$.
\end{enumerate}
\begin{proof}
We first obtain Part (i). Assumption \ref{a:regcis} Part (i) implies that $W_k(x_{1:N};\xi)\ge w_k(x_k;\xi)/CN$. Hence, for $k\ne l,$
    \begin{align*}
   (CN)  P(y,l;\dd{z}\times\{k\}) &\ge \int_{A^{N+1}}\frac{\pi(\dx_k)\eta(\xi|x_k)}{q_k(\dx_k|\xi)}T(x_k,\dd{z};\xi)\Gamma^N(\dx_{1:N},\dxi|y,l)\\
    &=\pi(\dd{z}) \int_{A^{N+1}}\frac{\eta(\xi|z)}{q_k(\dx_k|\xi)}T(z,\dd{x_k};\xi)\Gamma^N(\dx_{1:N},\dxi|y,l)\\
    &= \pi(\dd{z})\int_{A^{N+1}}\frac{\eta(\xi|z)\eta(\dxi|y)}{q_k(\dx_k|\xi)q_l(\dx_l|\xi)}T(z,\dx_k;\xi)T(y,\dx_l;\xi) \bs{q}(\dx_{1:N}|\xi)\\
    &=\pi(\dd{z})\int_{A^3}\frac{\bs{q}_{l,k}(\dx_l,\dx_k|\xi)}{q_k(\dx_k|\xi)q_l(\dx_l|\xi)}T(z,\dx_k;\xi)T(y,\dx_l;\xi) \eta(\xi|z)\eta(\dxi|y)\\
    &= \pi(\dd{z})\int_A\eta(\xi|z)\eta(\dxi|y)
     \times\left[\int_{A^2}\frac{\bs{q}_{l,k}(\dx_l,\dx_k|\xi)}{q_k(\dx_k|\xi)q_l(\dx_l|\xi)}T(z,\dx_k;\xi)T(y,\dx_l;\xi)\right]\\
    &=\pi(\dd{z})\int_A\eta(\xi|z)\eta(\xi|y)S_{l,k}(\xi,y,z)\mu(\dxi)\\
    &\ge \pi(\dd{z}) h_{l,k}(y,z).
    \end{align*}
    We can similarly result for $k = l$.
    We now prove part (ii). By part (i), Eq.~\eqref{eq:lbndP} holds for $t=1$. Suppose that \eqref{eq:lbndP} holds for some $t$. Then,
    \begin{align*}
    P^{t+1}(y,l;\dd{z}\times\{k\}) &= \sum_{j=1}^N\int_AP^t(y,l;\dd{v}\times\{j\})P(v,j;\dd{z}\times\{k\})\\
    &\ge \left(\frac{1}{C}\right)^t\sum_{j=1}^N\int_A\frac{\pi(\dd{v})}{N}H_{l,j}^t(y,v)\times \frac{\pi(\dd{z})}{CN}h_{j,k}(v,z)\\
    &=\left(\frac{1}{C}\right)^{t+1}\,\frac{\pi(\dd{z})}{N}\,\frac{1}{N}\sum_{j=1}^N\int_A\pi(\dd{v})H^t_{l,j}(y,v)h_{j,k}(v,z)\\
    &=\left(\frac{1}{C}\right)^{t+1}\,\frac{\pi(\dd{z})}{N}\,\E_{N^{-1}\pi}[H^t_{l,J}(y,V)h_{J,k}(V,z)]\\
    &= \left(\frac{1}{C}\right)^{t+1}\,\frac{\pi(\dd{z})}{N}\,H^{t+1}_{l,k}(y,z).
    \end{align*}
    Hence, the bound holds for all $t$.
    We now prove Part~(iii). We first show that $H^t_{l,k}(y,z)>0$ for $t =  2$ and then, recursively, for all $t\ge 2$.
   For any pair $y,z\in A$, and $l,k\in\{1{:}N\}$,
    \[
    H_{l,k}^2(y,z) = \E_{N^{-1}\pi}[h_{l,J}(y,V)h_{J,k}(V,z)]\ge \E_{N^{-1}\pi}[h_{l,J}(y,V)h_{J,k}(V,z)I(J\in \mathcal{J}_l\cap\mathcal{J}_k)] >0,
    \]
    where the last inequality follows from Assumption \ref{a:regcis} Part (ii). If $H_{l,j}^t(y,\cdot)>0$, then
    \[
    H_{l,k}^{t+1}(y,z) = \E_{N^{-1}\pi}[H_{l,J}^t(y,V)h_{J,k}(V,z)] \ge \E_{N^{-1}\pi}[H_{l,J}^t(y,V)h_{J,k}(V,z)I(J\in\mathcal{J}_k)] >0.
    \]
\end{proof}
\end{lemma}

\begin{proof}[Proof of Theorem \ref{thm:convergence}]
	The sequence $\{(y^{(t)},k^{(t)})\}$ from the MIIS algorithm is Markov, because the MIIS algorithm is a two component Gibbs sampler, and has transition kernel
	\begin{equation}
		P(y,l;B\times\{k\}) = \int_{A^{N+1}}W_k(x_{1:N},\xi)T(x_k,B;\xi)\Gamma^N(\dx_{1:N},\dxi|y,l).
		\label{eq:pmiis}
	\end{equation}

    The proof shows that for all starting values $(y,l)\in A\times\{1{:}N\}$, the $t^{th}$ step Markov transition kernel $P^t(y,l;B\times\{k\})$ is positive for all $t\ge 2$, and any $B\in\Omega$ such that $\pi(B) > 0 $ and $k\in\{1{:}N\}$.

    Suppose that $y \in A$, $B \in \Omega $ and $k,l \in \{1{:}N\}$. If $\pi(B) = 0 $ then $P^t(y,l;B\times\{k\}) = 0 $ for $t \geq 1$; if $\pi(B) > 0 $ then $P^t(y,l;B\times\{k\})>0$ for all $t\ge 2$ by Lemma ~\ref{lemma: prel to thm 5}.
 This  means that the marginal chain is $N^{-1}\pi$-irreducible and aperiodic and that $P(y,l;\dd z\times\{k\})$ is absolutely continuous with respect to $N^{-1}\pi(\dd z)$.
  It then
 follows from Theorem 1 and Corollary~1 in \citet{tierney1994} that for all $(y,l)\in A\times\{1{:}N\}$, $\lim_{t\rightarrow\infty}|P^t(y,l;\cdot-N^{-1}\pi(\cdot)|_{TV} = 0$, proving the first part of the theorem.
Proof of second part. Define $g_l(z): = \min_{k \in \mathcal{J}_l}\underbar{h}_{l,k}(z)  > 0 $. Then,
\begin{align*}
H^2_{l,k}(y,z) & = \sum_{k^\prime \in  \mathcal{J}_l}\int \pi(\dd{z}^\prime)  h_{l,k^\prime}(y,z^\prime) h_{k^\prime,k}(z^\prime, z)
 \geq \Big ( \int  \pi(\dd{z}^\prime) g_l(z^\prime) \Big )
 \sum_{k^\prime \in \mathcal{J}_l} \underbar{h}_{k^\prime,k}( z).
 \end{align*}

Let $D_1:= \int \pi(\dd{z}^\prime) g_l(z^\prime)$,
$$ D_2 :=   \sum_{k^\prime \in \mathcal{J}_l} \int  \underbar{h}_{k^\prime,k}( z)\pi(\dd{z}) \quad \mathrm{and}\quad
\nu(B) := D_2^{-1}\sum_{k^\prime \in \mathcal{J}_l}\int_B  \underbar{h}_{k^\prime,k}( z)\pi(\dd{z}).  $$
Then, from \eqref{eq:lbndP},
\begin{align*} P^2(y,l; \dd{z},\{k\}) & \geq C^{-2}D_1 D_2 N^{-1}\nu(\dd{z}) .
\end{align*}
and uniform ergodicity follows from Proposition 2 in \citet{tierney1994}.
\end{proof}

\subsection{Convergence of the MIIS Gibbs Sampler}
We again consider the marginal chain $\{y_t, \bs{l}_t, t \geq 0\} $ of the MIIS sampler, where $\bs{l}_t:= (l_{1:d})_t$.
Let $P_{s,M}(y(s),l_s; \dd{z(s)}\times \{k_s\}|y(\sk[s])) $ be the transition kernel for the s$^{th}$ component of the marginal chain. The transition kernel for the marginal chain is
\begin{align*}
P_M(y,\bs{l};\dd{z} \times \{ \bs{k}\}) = \prod_{s=1}^d P_{s,M}(y(s),l_s; \dd z(s) \times \{k_s\}|z({<s}),y({>s})),
\end{align*}
where we use the shorthand notation $z({<s}) = z(1{:}s-1)$ and $y(>s) = y(s+1{:}d)$. Define
\begin{align}
h_{\bs{l},\bs{k}}(y,z): = & \prod_{s=1}^d h_{s,l_s,k_s}(y(s),z(s);z({<s}),y({>s})).
\end{align}
We require the definition of the sub-stochastic kernel $H_{\bs l,\bs k}(y,\dd{z}) = C^{-1}N^{-d} h_{\bs l,\bs k}(y,z)P_G(y,\dd z)$ and, iteratively,
\begin{align*}
    H^{t+1}_{\bs{l},\bs{k}}(y,\dd z) &= \frac{1}{CN^d}\sum_{\bs j\in\{1:N\}^d}\int_AH^t_{\bs l, \bs j}(y,\dd v)\,h_{\bs j, \bs k}(v,z)P_G(v,\dd{z})\\
    &= \frac{1}{CN^d}\sum_{\bs j\in\{1:N\}^d}\int_Ah_{\bs l, \bs j}(y,v)\,H^t_{\bs j, \bs k}(v,\dd z)P_G(y,\dd{v})\\
    &=\E_{P_G(y,\cdot)/N^d}\left[h_{\bs l, \bs J}(y,V)\,H^t_{\bs J, \bs k}(V,\dd z)\right].
\end{align*}

\begin{lemma} \label{lemma: prel to thm 6}
Suppose Assumption~ \ref{a:regciscond} holds. Then,
\begin{enumerate}
\item [(i)]
The marginal chain $\{y^(t), \bs{l}^{(t)}, t \geq 0\} $ is Markov.
\item [(ii)] For $t=1,2,...$
\begin{align*}
 P_M^t(y,\bs{l};\dd{z}\times\{\bs{k}\}) &\ge  H^{t}_{\bs l,\bs z}(y,\dd z).
 \end{align*}
 \item [(iii)] Suppose $t\ge 2$, $B\in\Omega$, and $\pi(y)>0$. If $P_G^t(y,B)>0$ then $H^t_{\bs{l},\bs{k}} (y,B) > 0 $.
\end{enumerate}
\end{lemma}
\begin{proof}
Part~(i) follows from the construction of the MIIS sampler.

We show part (ii) by induction. By part (i) of Lemma~\ref{lemma: prel to thm 5}, for each $s=1,...,d$,
\begin{align*}
P_{s,M}(y(s),l_s; k_s,\dd{z}(s)|y(\sk[s])) & \geq C^{-1/d} N^{-1} h_{s,l_s,k_s}(y(s),z(s); y(\sk[s]))\pi_s(\dd{z(s)}|y(\sk[s])).
\end{align*}
Hence, for $t=1$, part~(ii) follows form the definition of $P_M(y,\bs l;z\times \{\bs k\})$ and $H_{\bs l,\bs k}(y,\dd z)$. Suppose  $P_M^t(y,\bs{l};\dd{v}\times\{\bs{j}\}) \ge H^{t}_{\bs l,\bs j}(y,\dd v)$, for $\dd v\in \Omega$ and $\bs j \in \{1{:}N\}^d$. Then
\begin{align*}
P_M^{t+1}(y,\bs{l};\dd{z}\times\{\bs{k}\}) &= \sum_{\bs j\in\{1{:}N\}^d}\int_A P_M^t(y,\bs{l};\dd{v}\times\{\bs{j}\})P_M(v,\bs{j};\dd{z}\times\{\bs{k}\})\\
&\ge C^{-1}N^{-d}\int_A \sum_{\bs j\in\{1{:}N\}}h_{\bs j,\bs k}(v,z)\,H^{t}_{\bs l,\bs j}(y,\dd v)P_G(v,\dd z)\\
& = H^{t+1}_{\bs l,\bs z}(y,\dd z).
\end{align*}
Then part (ii) also holds for $t+1$, proving the result.

Part (iii) follows By induction. We first show that the result holds for $t=2$ and, then, we show that if the result holds for some $t\ge 2$, it also holds for $t+1$. Let $\mathcal{J}_{\bs l} = \times_{s=1}^d \mathcal{J}_{s,l_s}$ and verify that, under assumption \ref{a:regciscond}(ii) part b, $\mathcal{J}_{\bs l}\cap \mathcal{J}_{\bs k} \ne \emptyset$, for any pair $\bs l,\bs k\in\{1{:}N\}$.

Suppose $t=2$. If $P_G^2(y,B)>0$, there is a set $F'\in\Omega$ such that $P_G(y,F')>0$ and $P_G(x,B)>0$ for $x\in F'$. Let $F'\supseteq F = F_1\times\cdots\times F_d$. For $v\in F$ (i.e., each $v(s)\in F_s$), $s=1,\ldots,d$, and $j$ in $\mathcal{J}_{s,l}\cap\mathcal{J}_{s,k}$,
\begin{equation*}
h_{s,l,j}(y(s),v(s);v({<s}),y({>s}))h_{s,j,k}(v(s),z(s);z({<s}),v({>s})) > 0,
\end{equation*}
from Assumption \ref{a:regciscond}(ii) part (c). Therefore,
\begin{multline*}
\sum_{\bs j\in\{1{:}N\}}h_{\bs l, \bs j}(y,v)h_{\bs j,\bs k}(v,z) \ge
\prod_{s=1}^d\sum_{j\in\mathcal{J}_{s,k_s}\cap\mathcal{J}_{s,l_s}}h_{s,l_s,j}(y(s),v(s);v({<s}),y({>s}))\\
\,\times h_{s,j,k_s}(v(s),z(s);z({<s}),v({>s})) > 0
\end{multline*}
for $v\in F$, and any $\bs l, \bs k\in\{1{:}N\}^d$. Hence
\begin{align*}
H^2_{\bs l, \bs k}(y,z) & = \frac{1}{C^2}\int_B\int_A\sum_{\bs j\in\{1{:}N\}}h_{\bs l, \bs j}(y,v)h_{\bs j,\bs k}(v,z)P_G(y,\dd v)P_G(v,\dd z)\\
&\ge \frac{1}{C^2}\int_B\left\{\int_F\sum_{\bs j\in\mathcal{J}_{\bs l}\cap\mathcal{J}_{\bs k}}h_{\bs l, \bs j}(y,v)h_{\bs j,\bs k}(v,z)P_G(y,\dd v)P_G(v,\dd z)\right\} > 0,
\end{align*}
where the last line follows from calculating each integral between brackets over $F_s$, $s=1,\cdots,d$.

Suppose that part (iii) holds for some $t\ge 2$ and that $P_G^{t+1}(y,B)>0$. Then,
\begin{align*}
   H^{t+1}_{\bs l,\bs k}(y,B) &  =  \frac{1}{CN^d}\sum_{\bs j\in \{1{:}N\}}\int_B\int_A H^t_{\bs l,\bs j}(y,\dd v)h_{\bs j,\bs k}(v,z) P_G(v,\dd z)\\
   & \ge \frac{1}{CN^d}\int_B\int_{F}\left[\sum_{\bs j\in \mathcal{J}_{\bs k}} H^t_{\bs l,\bs j}(y,\dd v)h_{\bs j,\bs k}(v,z)\right] P_G(v,\dd z)> 0,
\end{align*}
where $F\in\Omega$ is such that $P_G(x,B)>0$ for $x\in F$ and $P_G^t(y,F)>0$. The result holds for any $\bs l$ and $\bs k$ in $\{1{:}N\}^d$.
\end{proof}

\begin{proof} [Proof of Theorem~\ref{thm:convergencecond}]
   The result follows from Lemma \ref{lemma: prel to thm 6} and Theorem 1 in \cite{tierney1994}. First define the Markov kernel $[N^{-d}P_G](y,\bs l;B\times\{\bs k\})$, that is the kernel of the Gibbs sampler that draws $(z(s),k_s)|(z(<s),y(>s),\bs k_{<s},\bs l_{>s})$ from $N^{-1}\pi_s(z(s)|z(<s),y(>s))$, sequentially. If the Gibbs kernel $P_G$ is irreducible and aperiodic, so it is the kernel $[N^{-d}P_G]$, since all $k_s\in\{1{:}N\}$, $s=1,\dots,d$, are accessible at each iteration. The proof consists in showing that accessible sets from $[N^{-d}P_G]^t$, the ideal Gibbs in $t\ge2$ steps, are also accessible by the MIIS-Gibbs kernel after $t$ iterations, $P_M^t$. Lemma \ref{lemma: prel to thm 6} (i) shows that $P_M$ is a Markov kernel. Parts (ii) and (iii) together show that $P_G(y,B)>0$ implies that $P_M(y,\bs k;B\times\{\bs l\})>0$ for any pair $(\bs l,\bs k)$. Hence, all sets accessible by $[N^{-d}P_G]$ are also accessible by $P_M$, which implies that $P_M$ is also irreducible. To show that $P_M$ is aperiodic, we assume by contradiction that $P_M$ is not aperiodic. In this case, $[N^{-d}P_G]$ would have to be periodic as well, which contradicts with the assumption that $[N^{-d}P_G]$ is aperiodic. The result follows from Theorem 1 of \citet{tierney1994}.
\end{proof}

It also follows from Theorem \ref{thm:convergencecond} that, $\lim_{t \rightarrow \infty} P_M^t(y(s),l_s;\cdot|y(\sk[s])) = N^{-1}\pi_s(\cdot|y(\sk[s]) $, which implies that the control variates in Section \ref{s:CV} can be safely used.

\begin{proof}[Proof of Corollary~\ref{l:IS-CIScond}]
We can check that the conditions of Assumption 2 hold in a similar way to the proof of Corollary~\ref{C: miis sis}. The result follows from Theorem~\ref{thm:convergencecond}.
\end{proof}

\subsection{Proofs of consistency}\label{s:VR theory proofs}
\begin{proof}[Proof of Corollary~\ref{thm:MIISconsistent}]
 The distribution of $\{l^{(t)},y^{(t)}\}$ converges to $N^{-1}\pi(\cdot)$ by Theorem \ref{thm:convergence}. Let
${\wh E}_{CIS,t}^N(f)$ be defined by Equation~\eqref{eq:pimiis} in Section~\ref{SS: reusing}. The result now follows from
Lemma~\ref{thm:CISunbiased}, which shows that each ${\wh E}_{CIS,t}^N(f)$ is unbiased and by
the strong law of large numbers for ergodic sequences \citep[Theorem~3]{tierney1994}.
\end{proof}

\begin{proof}[Proof of Corollary~\ref{thm:RBconsistent}]
The distribution of $(l_{1:d}^{(t)},y^{(t)})$ converges to $N^{-d}\pi(\cdot)$ by Theorem~\ref{thm:convergencecond}
The result follows from Lemma~\ref{thm:CISunbiased-conditional} and the strong law of large numbers for ergodic sequences \citep[Theorem~3]{tierney1994}.
\end{proof}

\begin{proof}[Proof of Corollary~\ref{thm:CVconsistent}]
For any $f$ with $E_\pi(|f|)<\infty$, it follows from Corollary~\ref{thm:MIISconsistent}
that ${\wh E}_{MC}^M(f)\rightarrow\pi(f)$, and ${\wh E}_{MIIS}^{M,N}(f)\rightarrow\pi(f)$ with probability one.
This means that ${\wh E}_{MC}^M(f)-{\wh E}_{MIIS}^{M,N}(f)\rightarrow 0$, with probability one. Hence, for any constant $\boldsymbol\kappa\in\mathbb{R}^p$, and $\pi$-integrable functions $g_1,\dots,g_p$, the linear combination $\sum_{i=1}^p \kappa_i[\pi_{MC}^M(g_i)-\pi_{MIIS}^{M,N}(g_i)] \rightarrow 0$ with probability one. The proof now follows from
Corollary~\ref{thm:MIISconsistent}
\end{proof}

\begin{proof}[Proof of Corollary~\ref{thm:CVconsistentRB}]
The proof of this corollary follows the same arguments used in the proof of Corollary~\ref{thm:CVconsistent}.
\end{proof}

\end{document}